\documentclass{vldb}

\usepackage{booktabs} 
\usepackage{amssymb,amsmath,bbm,epsfig,graphics,enumerate,float,xspace}
\usepackage{color,graphicx,url,booktabs}
\usepackage{subfig,balance,algorithm,algorithmic,lineno,multirow}
\usepackage{soul,url}
\usepackage{setspace}
\newtheorem{theorem}{Theorem}[section]
\newtheorem{lemma}[theorem]{Lemma}
\newtheorem{definition}[theorem]{Definition}
\newtheorem{assumption}[theorem]{Assumption}

\begin{document}

\title{Towards Differentially Private Truth Discovery for Crowd Sensing Systems \titlenote{This work was done when the authors Yaliang Li, Houping Xiao and Zhan Qin were at State University of New York at Buffalo.}}

\numberofauthors{1} 
\author{
\alignauthor Yaliang Li{\small$~^{1}$}, Houping Xiao{\small$~^{2}$}, Zhan Qin{\small$~^{3}$}, Chenglin Miao{\small$~^{4}$}, \\ Lu Su{\small$~^{4}$}, Jing Gao{\small$~^{4}$}, Kui Ren{\small$~^{4}$}, and Bolin Ding{\small$~^{5}$} \\
\affaddr{$^{1}$\,Tencent Medical AI Lab, Palo Alto, CA USA}  \\
\affaddr{$^{2}$\,Georgia State University, Atlanta, GA USA}  \\
\affaddr{$^{3}$\,University of Texas at San Antonio, San Antonio, TX USA} \\
\affaddr{$^{4}$\,State University of New York at Buffalo, Buffalo, NY USA}  \\
\affaddr{$^{5}$\,Alibaba Group, Bellevue, WA USA} \\
\affaddr{$^{1}$\,yaliangli@tencent.com, $^{2}$\,hxiao@gsu.edu, $^{3}$\,zhan.qin@utsa.edu}  \\
\affaddr{$^{4}$\,\{cmiao, lusu, jing, kuiren\}@buffalo.edu, $^{5}$\,bolin.ding@alibaba-inc.com}
}

\maketitle

\begin{abstract}
Nowadays, crowd sensing becomes increasingly more popular due to the ubiquitous usage of mobile devices. However, the quality of such human-generated sensory data varies significantly among different users. To better utilize sensory data, the problem of truth discovery, whose goal is to estimate user quality and infer reliable aggregated results through quality-aware data aggregation, has emerged as a hot topic.  Although the existing truth discovery approaches can provide reliable aggregated results, they fail to protect the private information of individual users. Moreover, crowd sensing systems typically involve a large number of participants, making encryption or secure multi-party computation based solutions difficult to deploy. To address these challenges, in this paper, we propose an efficient privacy-preserving truth discovery mechanism with theoretical guarantees of both utility and privacy. The key idea of the proposed mechanism is to perturb data from each user independently and then conduct weighted aggregation among users' perturbed data. The proposed approach is able to assign user weights based on information quality, and thus the aggregated results will not deviate much from the true results even when large noise is added. We adapt local differential privacy definition to this privacy-preserving task and demonstrate the proposed mechanism can satisfy local differential privacy while preserving high aggregation accuracy. We formally quantify utility and privacy trade-off and further verify the claim by experiments on both synthetic data and a real-world crowd sensing system.
\end{abstract}

\section{Introduction}
\label{sec:introduction}

Today, we witness the explosion of sensory data which are continuously generated by countless individuals all over the world through the increasingly capable and affordable mobile devices, such as smartphones, smartwatches, and smartglasses. The information mined from such massive human-generated sensory data provides critical insights for a wide spectrum of applications, including healthcare, smart transportation and many others. While sensory data are potentially huge treasure troves, it remains a challenging task to extract truthful information from the noisy, conflicting and heterogeneous data submitted by the numerous mobile device users.

In such scenarios, it is essential to aggregate the sensory data about the same set of objects collected from a crowd of users to get true facts or aggregated results. The key factor in aggregating noisy sensory data is to capture the difference in information quality among different users. Some users provide correct and useful information while others may submit noisy or fake information due to hardware quality, environment noise, or even the intent to deceive and get rewards. Therefore, the naive approach that regards all the users equally in aggregation may fail to derive reliable aggregated results. Instead, we hope to capture the probability of a user providing accurate information in the form of user weight and incorporate it into the aggregation so that final output is closer to the information provided by reliable users. The challenge is that user weight is usually unknown \emph{a priori} in practice and has to be inferred from the sensory data.

To address this challenge, a series of \emph{truth discovery} mechanisms \cite{yin_kdd07,luna_survey_vldb12,li2015survey} are proposed to tackle the problem of estimating user weight and inferring reliable aggregated information from noisy crowdsourced sensory data, and have been successfully applied to various domains such as social sensing \cite{WKL+12}, air quality monitoring \cite{meng2015truth}, and network quality measurement \cite{liweighted}. In these applications, users can share their sensory data, and an accurate aggregation can lead to important knowledge for various applications and systems. As both user weights and aggregated results are unknown, truth discovery approaches estimate them simultaneously based on the following two principles: (1) If the information provided by a particular user is closer to the aggregated results, this user will be assigned a higher weight. (2) If a user has a higher weight, his information will be counted more in the aggregation. Based on these principles, truth discovery approaches take crowdsourced sensory data as input, and then iteratively estimate user weights and update aggregated results. Different from simple aggregation such as averaging or voting, truth discovery conducts weighted aggregation in which user weights are automatically estimated from the sensory data.

\textbf{Privacy Concerns}. One important component missing in these truth discovery approaches is the protection of user privacy. These approaches assume that the sensory data has been collected from users by a centralized server. However, during this data collection procedure, users may have concerns in sharing personal or sensitive information \cite{ganti2008poolview,gulisano2016bes,kayes2015privacy,pyrgelis2016privacy,agir2014user,drosatos2014privacy,pournajaf2016participant,xue2011distributed}. For example, individuals' GPS data are important sources for traffic monitoring and smart transportation, but contain sensitive information that users might not want to release. Aggregating health data through wearable devices can lead to better discovery of new drugs' effects, but it may suffer from the risk of information leaking about each participant. In these and many more application scenarios, the final aggregation results can be public and beneficial to the community or society, but the data from each individual user should be well protected.

A possible solution to tackle this challenge is to adopt encryption or secure multi-party computation techniques in truth discovery \cite{miao2015cloudpp,miao2017privacy,zheng2017privacy,zheng2018learning}. However, these techniques typically involve time-consuming computation or expensive communication among mobile device users. Therefore, although these techniques can achieve strong protection for users, it is difficult to deploy them in large-scale truth discovery tasks which require highly efficient and non-coordinated privacy preserving strategies.

\textbf{Proposed Mechanism}. In the light of these challenges, we propose to address such user privacy concerns in truth discovery by providing an effective and efficient mechanism. The proposed privacy-preserving truth discovery mechanism for aggregating sensory data can guarantee both good utility and strong privacy. The proposed mechanism works as follows. Each user samples independent noise from a privately known noise distribution and adds the sampled noise to their data. After collecting perturbed data from all the users, the server conducts weighted aggregation by weighing each user's information properly to obtain final output. We demonstrate the ability of the proposed mechanism to tolerate high noise with negligible loss in aggregation accuracy. This advantage is brought by the fact that the proposed mechanism can automatically adjust user weights, and thus can lower a user's weight when high noise is added so that the effect of noise on the final aggregation results can be significantly reduced.

The theoretical analysis of the proposed mechanism is conducted from the following aspects: (1) We quantify the loss in aggregation accuracy that is caused by the noise added to the input data, and show the proposed mechanism's advantage that the accuracy does not drop much even with large noise. (2) Another advantage of the mechanism is that the noise distribution adopted by each user is unknown to the public. This scheme is easy to implement and requires no communication among users. Formally, we adopt local differential privacy to quantify user privacy protection in truth discovery scenarios. (3) The trade-off between aggregation accuracy (utility) and the defined local differential privacy is analyzed, which shows how both aggregation accuracy and end user privacy can be guaranteed simultaneously.

\textbf{Contributions}. In summary, our contributions are:
\begin{itemize}
\item We propose a privacy-preserving truth discovery mechanism for crowdsourced sensory data aggregation, which consists of perturbations on users' data and weighted aggregation on perturbed data. The proposed mechanism tackles this challenging privacy preserving task with guarantees of both accuracy and privacy.
\item We formally define aggregation accuracy and privacy for the studied task, and theoretically quantify the range of noise that can be adopted to achieve good utility and strong privacy.
\item Experiments on both synthetic data and a real-crowd sensing system validate the claim that the proposed mechanism can generate accurate aggregation results while preserving users' privacy. Results show that even when the added noise is large, aggregation accuracy only drops slightly. 
\end{itemize}

In the remaining parts of this paper, we first define the problem in Section \ref{sec:problem}. Then the proposed privacy-preserving truth discovery mechanism is presented in Section \ref{sec:method}. Section \ref{sec:analysis} theoretically analyzes the utility and privacy trade-off, which is also demonstrated through a series of experiments in Section \ref{sec:exp}. We discuss related work in Section \ref{sec:related} and conclude the paper in Section \ref{sec:conclusion}.

\section{Problem Definition}
\label{sec:problem}

In this section, we describe the setting of the proposed privacy preserving task based on crowd sensing system. At the core, it consists of two parties: \emph{server} and \emph{users}. The server is a data collector and computation platform, which is used to collect and aggregate sensory data from a crowd of users. Users represent the participants of the task, who are usually driven by their interests or financial incentives. They receive assigned tasks from the server and submit their sensory data to the server.

We propose to protect \emph{users}' privacy in the sense that users' data are obfuscated before being submitted to the untrusted server. Providing privacy protection for the users who submit data to an untrusted server is essential in crowd sensing system. With an effective privacy protection mechanism, users are more confident and willing to share data, which greatly enhances data collection and enables crowd sensing tasks that would otherwise be infeasible due to privacy concerns.

The security threats in the crowd sensing system mainly come from the unfaithful behavior of the server as it can tamper the confidentiality of users' provided information. The server might try to deduce extra knowledge about users due to curiosity or financial incentives. The users' security concern is to protect their private sensory data from leaking out, while enabling the server to execute aggregation over them. The formal definition is introduced below.

\emph{Problem Definition}. Suppose there are $N$ objects (i.e., micro-tasks) that the server wants to collect information about, and there are $S$ users to provide information about these objects. Let continuous data $x_n^s$ represent the information for the $n$-th object provided by the $s$-th user. Instead of submitting their original data $\{x_n^s\}_{n,s=1}^{N,S}$ to the server, each user perturbs his data and only the perturbed data $\{\hat{x}_n^s\}_{n,s=1}^{N,S}$ are submitted. Our goal is to protect users' privacy by making the probability of observing the same perturbed value given different original values $P(\hat{x}_1=\hat{x}_2|x_1\neq x_2)$ high, while keeping the aggregation on $\{\hat{x}_n^s\}_{n,s=1}^{N,S}$ close enough to the true aggregated values.
Figure \ref{example} illustrates this task setting.

\begin{figure}[tbh]
\centering
\includegraphics[width=0.42\textwidth]{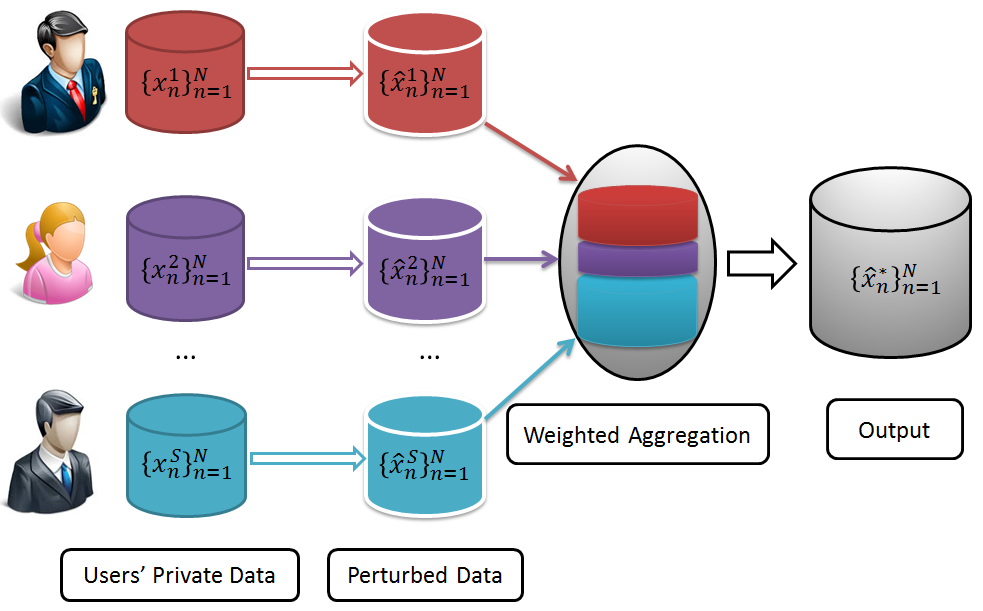}
\caption{Privacy-Preserving Truth Discovery}
\label{example}
\end{figure}

Note that there are certainly other security threats coming from inside or outside of crowd sensing systems. For the other threats, we can leverage and integrate existing techniques to make our model more complete and readily being deployed in real world systems.

\section{Methodology}
\label{sec:method}

In this section, we first introduce the concepts of truth discovery, and then present the proposed privacy-preserving truth discovery mechanism that can guarantee both good utility and strong privacy.

\subsection{Truth Discovery}
\label{sec:preliminary}

In crowd sensing systems, multiple observations are provided by different users on the same set of objects. However, the quality of user-provided information usually varies a lot across users. Therefore, the naive approach that treats all the users equally in aggregation may fail to give reliable aggregated results. Truth discovery \cite{yin_kdd07,luna_survey_vldb12,li2015survey} gains increasing popularity recently as it can infer user weights and conduct weighted aggregation on multiple noisy data sources. Instead of regarding all the users equally, truth discovery approaches estimate users' information quality from the data and relies more on the users who contribute high-quality information to derive aggregated results. 

Although existing truth discovery approaches may differ in the specific ways to compute aggregated results and user weights, we summarize their common procedure as follows. As there are $S$ users providing their information, the goal is to aggregate $\{x_n^s\}_{s=1}^S$ to infer the reliable information about the $n$-th object, $x_n^*$. Note that we assume both input and output are continuous values. The general procedure of truth discovery is summarized in Algorithm \ref{alg:truth_discovery}. Truth discovery starts with an initialization of user weights, and then iteratively conducts aggregation step and weight estimation step until convergence. The convergence criterion can be a threshold for the change of the aggregated results in two consecutive iterations or a predefined iteration number.

\smallskip \noindent \textsf{\textbf{Aggregation}}. In the aggregation step, the user weights are fixed. Then we infer aggregated results as follows:
\begin{equation}
x_n^* = \frac{\sum_{s=1}^{S} w_s\cdot x_n^s}{\sum_{s=1}^{S} w_s},
\label{eq:weighted_mean}
\end{equation}
where $w_s$ is the weight of the $s$-th user. In this weighted aggregation framework, the final result $x_n^*$  relies on those users who have high weights. This follows the principle that the information from reliable users will be counted more.

\smallskip \noindent \textsf{\textbf{Weight Estimation}}. In this step, user weights are inferred based on the current aggregated results. The basic idea is that if a user provides information which is close to the aggregated results, a high weight will be assigned to this user. Typically, user weights are calculated as follows:
\begin{equation}
w_s = f\left(\sum_{n=1}^N d(x_n^s,x_n^{*})\right),
\label{eq:weight_estimation}
\end{equation}
where $d(\cdot)$ is a function that measures the difference between the user-provided information and the aggregated results, and $f$ is a monotonically decreasing function. If the difference is small, then the user gets a high weight. Different truth discovery methods may adopt various functions $d(\cdot)$ and $f$, but the underlying principle is the same.

Here we show the weight estimation of CRH \cite{crh_sigmod14} as an instantiation of Eq.\ (\ref{eq:weight_estimation}):
\begin{equation}
w_s = -\log \left(\frac{\sum_{n=1}^{N} d(x_n^s,x_n^{*})}{\sum_{s'=1}^{S}\sum_{n=1}^{N} d(x_n^{s'},x_n^{*})}\right).
\label{eq:weight_CRH}
\end{equation}
Note that the proposed privacy preserving mechanism is not specifically designed for CRH. It can work with any truth discovery method that can handle continuous data. In Section \ref{sec:exp}, we will demonstrate experimental results that support this claim.

\begin{algorithm}[tbp]
\small
\caption{\bf Truth Discovery}\label{alg:truth_discovery}
\flushleft
{\bf Input:} Information from $S$ users $\{x_n^s\}_{n,s=1}^{N,S}$. \\
{\bf Output:} Aggregated results $\{x_n^{*}\}_{n=1}^{N}$.
\begin{algorithmic}[1]
\STATE Initialize the user weights $\{w_s\}_{s=1}^{S}$;
\REPEAT
\FOR{$i$ $\leftarrow$ $1$ to $N$}
    \STATE According to Eq.\ (\ref{eq:weighted_mean}), update aggregated results based on the current estimation of user weights;
\ENDFOR
\STATE According to Eq.\ (\ref{eq:weight_estimation}), estimate user weights based on the current aggregated results;
\UNTIL{Convergence criterion is satisfied;}
 \RETURN Aggregated results $\{x_n^{*}\}_{n=1}^{N}$.
\end{algorithmic}
\end{algorithm}

\subsection{Proposed Mechanism}
The proposed privacy-preserving truth discovery mechanism consists of the following two components:

First, we propose to add i.i.d. Gaussian noise, $\xi_n^s$, to the original data provided by the $s$-th user on the $n$-th object, $x_n^s$. Let's denote the perturbed information as $\hat{x}_n^s$, then
\begin{equation}
\label{eq:add_noise}
\hat{x}_n^s = {x}_n^s + \xi_n^s,
\end{equation}
where $\xi_n^s \sim N(0,{\delta_s}^2)$. ${\delta_s}^2$ is the variance of the Gaussian noise chosen by the $s$-th user. Intuitively, the added noise is related to the degree of privacy protection. When noise variance is large, the added noise is more likely to be large and more privacy protection is expected. To guarantee aggregation accuracy, we ask each user to sample his own variance from an exponential distribution with hyper parameter $\lambda_2$. Based on this strategy, each user chooses his noise variance independently and then sample independent noise from his private noise distribution. 

After data are perturbed, each user submits his perturbed data $\{\hat{x}_n^s\}_{n=1}^{N}$ to the server. The server aggregates the perturbed data from all the users $\{\hat{x}_n^s\}_{n,s=1}^{N,S}$ by conducting truth discovery to obtain final output for all objects. When aggregating perturbed data by truth discovery, the weight of each user is estimated based on the quality of information after perturbation. By conducting weighted aggregation, the effect of noise will be characterized in the user weights and the final results would not deviate much from the aggregated results without perturbation. This promises good utility of the aggregated results. The whole procedure is illustrated by the following example and summarized in Algorithm \ref{alg:proposed}.

\emph{Example: Consider a user who has high-quality original data about the objects of interest. Following the proposed mechanism, this user will sample his own variance, say a large one, to perturb his original data. The perturbed data is submitted to the server, and aggregated by truth discovery. From the perspective of privacy, any other parties including the server do not know this particular user's original data and his sampled variance, thus the privacy guarantee is provided. From the perspective of utility, the estimated weight of this particular user will be low as the quality of his perturbed data is not good. Thus the aggregated results will not be affected too much, and good utility can be guaranteed.}

\begin{algorithm}[htbp]
	\small
	\caption{\bf Privacy-preserving Truth Discovery}\label{alg:proposed}
	\flushleft
	{\bf Input:} $N$ objects (i.e., micro-tasks), $S$ users \\
	{\bf Output:} Aggregated results $\{\hat{x}_n^{*}\}_{n=1}^N$
	\begin{algorithmic}[1]
		\STATE Server sends micro-tasks to each user;
		\STATE Users finish the micro-tasks, i.e., the $s$-th user prepares his original information $\{x_n^s\}_{n=1}^N$;
		\STATE Each user samples his own parameter ${\delta_s}^2$ from exponential distribution based on the server-released hyper parameter $\lambda_2$;
		\STATE According to Eq.\ (\ref{eq:add_noise}), the $s$-th user perturbs his original information and get the perturbed data $\{\hat{x}_n^s\}_{n=1}^N$;
		\STATE Users submit their perturbed data to the server;
		\STATE Server conducts truth discovery on perturbed data $\{\hat{x}_n^s\}_{n,s=1}^{N,S}$ to calculate aggregated results.
		\RETURN Aggregated results $\{\hat{x}_n^{*}\}_{n=1}^N$.
	\end{algorithmic}
\end{algorithm}

Although the proposed mechanism is simple, it has several nice properties which make it a great choice for user privacy protection in truth discovery:
\begin{itemize}
\item  First, each user chooses his noise variance independently and randomly, so the noise distribution is unknown to any other parties including the server.
\item Second, truth discovery methods which conduct weighted aggregation make it possible to achieve high accuracy even when the added noise is large. This provides better accuracy than traditional aggregation methods, such as mean or median, which do not consider user weights based on information quality.
\item Last but not least, this technique ensures fast processing as each user only needs to generate random noise and add it to his data, and there are no communication costs due to the non-collaborative mechanism. It is easy to implement and use in real practice.
\end{itemize}

\section{Theoretical Analysis}
\label{sec:analysis}

In this section, we analyze the performance of the proposed privacy-preserving truth discovery mechanism from utility and privacy perspectives, quantify their trade-off, and demonstrate that the proposed mechanism can achieve good utility with strong privacy protection theoretically.

We first introduce the notations that will be used in the following analysis. As the proposed privacy-preserving truth discovery mechanism has two components, we denote the perturbation mechanism and the truth discovery algorithm as $\mathcal{M}$ and $\mathcal{A}$ respectively. The original data set is represented as $D=\{x^s_n\}_{n,s=1}^{N,S}$ in which $x^s_n$ is the original value contributed by the $s$-th user on the $n$-th object. The perturbed data set is denoted as $\mathcal{M}(D)=\{\hat{x}^s_n\}_{n,s=1}^{N,S}$ after following $\mathcal{M}$ to perturb $D$.  The outputs of the truth discovery algorithm on original data and perturbed data are denoted as $\{x^*_n\}_{n=1}^{N}=\mathcal{A}(D)$ and $\{\hat{x}^*_n\}_{n=1}^{N}=\mathcal{A}(\mathcal{M}(D))$ in which $x^*_n$ and $\hat{x}^*_n$ denote the aggregated result for the $n$-th object on original data and perturbed data respectively.

We also introduce an important parameter used in the following analysis. This parameter is related to the prior knowledge held by the server regarding noise. As discussed in the previous section, the noise added to the data follows Gaussian distribution $N(0,{\delta_s}^2)$ in which ${\delta_s}^2$ is drawn from an exponential distribution with parameter $\lambda_2$. The server does not know the actual noise distributions, but knows the hyper-parameter $\lambda_2$. In other words, the server knows the distribution for the variance that captures the noise distributions. Formally, we define prior knowledge as follows:
\begin{assumption}[Prior Knowledge]
Prior knowledge is the variance of the noise's distributions, i.e., the p.d.f of the noise's variance is $g(z)=\lambda_2 e^{-\lambda_2 z}$.
\label{assumption1}
\end{assumption}

Recall that in truth discovery, it is observed that the error in the original data (difference between user input and true aggregated results) follows Gaussian distribution $N(0,{\sigma_s}^2)$. If all the users are unreliable (with big ${\sigma_s}^2$), it is hard or even impossible to get useful aggregated results. Thus previous work on truth discovery assumes that most users should have relatively good quality \cite{bo_ltm_vldb12}. Following this, we assume that the error variance ${\sigma_s}^2$ is drawn from exponential distribution with parameter $\lambda_1$, which guarantees that the chance of observing an unreliable user is not very large. Note that parameter $\lambda_1$ is introduced only for theoretical analysis and is not involved in the proposed mechanism (Algorithm \ref{alg:proposed}). In practice, we do not need to estimate $\lambda_1$ for a given application.

Accordingly, the expectation of the error and noise's variances are  $1/\lambda_1$ and $1/\lambda_2$ respectively. Let $1/\lambda_2 = c/\lambda_1$. Then, $c$ stands for the ratio between the expectation of noise's variance and that of the error's variance. A large $c$ may lead to large noise added to users' data, and thus $c$ can be regarded as noise level compared with original data. $c$ is an important parameter. In the following analysis, we link utility and privacy to $c$ respectively and then discuss utility-privacy trade-off.

\subsection{Utility Analysis}

In this section, we present formal definition of utility, and analyze the utility of the proposed mechanism. First, we define the utility as follows.

\begin{definition}[$(\alpha,\beta)$-Utility]
Let $\beta\in[0,1]$ and $\alpha\geq0$. An algorithm $\mathcal{A}$ with perturbation mechanism $\mathcal{M}$ satisfies $(\alpha,\beta)$-Utility, if the following inequality holds:
\begin{align}
\Pr\{|\mathcal{A}(D)-\mathcal{A}(\mathcal{M}(D))|\geq\alpha\}\leq \beta,
\end{align}
where $D$ is an arbitrary data set.
This definition quantifies the probability of the difference in aggregation before and after perturbation. We hope that the chance of this difference is greater than $\alpha$ is smaller than probability $\beta$. Based on this definition, under perturbation mechanism $\mathcal{M}$, the smaller $\alpha$ and $\beta$ are, the better utility an algorithm $\mathcal{A}$ has.
\end{definition}

Let $\mathcal{A}$ be a truth discovery approach, and $\mathcal{M}$ be the perturbation approach in the proposed mechanism. We now quantify the utility of the proposed mechanism according to noise level $c$ and utility parameters $\alpha$ and $\beta$. The proof is derived based on the common property held by truth discovery approaches, i.e., weighted aggregation and weight estimation.

We derive the main result about utility shown in Theorem \ref{utility} when $c\neq 1$. The special case when $c=1$ is shown with similar results in Appendix A. We present the theorem and proof first and then discuss this result in detail.
\begin{theorem}
Consider a truth discovery algorithm $\mathcal{A}$. We apply perturbation mechanism $\mathcal{M}$ stated in Algorithm 2 on a data set $D$ and then apply truth discovery algorithm $\mathcal{A}$. Based on Assumption \ref{assumption1}, for the aggregation output before and after perturbation $\{x^*_n\}_{n=1}^{N}=\mathcal{A}(D)$ and $\{\hat{x}^*_n\}_{n=1}^{N}=\mathcal{A}(\mathcal{M}(D))$, there exist constants $\alpha_{\lambda_1,c}$ and $C_{\lambda_1,\alpha,\beta,S}$, s.t. $\forall \alpha > \alpha_{\lambda_1,c}$, $\beta\in[0,1]$, and $c\leq C_{\lambda_1,\alpha,\beta,S}$, $\mathcal{A}$ satisfies $(\alpha,\beta)$-Utility. Namely, the following inequality holds:
\begin{eqnarray}
\Pr\{\frac{1}{N}\sum_{n=1}^{N}|x^*_n-\hat{x}^*_n| \geq \alpha\}\leq \beta,
\end{eqnarray}
where $C_{\lambda_1,\alpha,\beta,S} = \lambda_1\sqrt{\pi}(\frac{\alpha^2\beta S^2}{4\sqrt{2}}+\frac{\alpha^2\sqrt{\pi}}{8}+\alpha+\frac{2}{\sqrt{\pi}})-2$ and $\alpha_{\lambda,c} = \frac{2\sqrt{2}}{\sqrt{\lambda_1}(1-c)}(\frac{3}{4}-\frac{c(c+\sqrt{c}+1)}{\sqrt{2}(1+\sqrt{c})})$.
\label{utility}
\end{theorem}
Before we give the proof of Theorem \ref{utility}, we first introduce a lemma that is useful to the proof.
\begin{lemma}
Assume $w_s=f(t_s)$ for all $s\leq S$. Provided $f$ is a monotonically decreasing function,  then we have:
\begin{small}
\begin{eqnarray}
\frac{\sum_{s=1}^{S}w_st_s}{\sum_{s=1}^{S}w_s}\leq\frac{\sum_{s=1}^{S}t_s}{S}.
\label{lemma1eq}
\end{eqnarray}
\end{small}\label{lemma1}
\end{lemma}
For the detailed proof of this lemma, please refer to Appendix B. Now we are ready to prove Theorem \ref{utility}.
\begin{proof}
\begin{align}
&\frac{1}{N}\sum_{n=1}^{N}|x^*_n-\hat{x}^*_n| \notag\\
= & \frac{1}{N}\sum_{n=1}^{N}|\frac{\sum_{s=1}^{S}w_sx^s_n}{\sum_{s=1}^{S}w_s} - \frac{\sum_{s=1}^{S}\hat{w}_s\hat{x}^s_n}{\sum_{s=1}^{S}\hat{w}_s}|\notag\\
=& \frac{1}{N}\sum_{n=1}^{N}|\frac{\sum_{s'\in S}\hat{w}_{s'}\sum_{s=1}^{S}w_sx^s_n-\sum_{s=1}^{S}w_s\sum_{s'\in S}\hat{w}_{s'}\hat{x}^{s'}_n}{\sum_{s=1}^{S}w_s\sum_{s'\in S}\hat{w}_{s'}}| \notag \\
=& \frac{1}{N}\sum_{n=1}^{N}|\frac{\sum_{s=1}^{S}\sum_{s'=1}^{S}\hat{w}_{s'}w_sx^s_n-\sum_{s=1}^{S}\sum_{s'=1}^{S}w_s\hat{w}_{s'}\hat{x}^{s'}_n}{\sum_{s=1}^{S}w_s\sum_{s'=1}^{S}\hat{w}_{s'}}|\notag\\
\leq& \frac{\sum_{s=1}^{S}\sum_{s'=1}^{S}\hat{w}_{s'}w_s(\frac{1}{N}\sum_{n=1}^{N}|x^s_n-\hat{x}^{s'}_n|)}{\sum_{s=1}^{S}\sum_{s'=1}^{S}\hat{w}_{s'}w_s}\notag\\
\leq&\frac{\sum_{s=1}^{S}\sum_{s'=1}^{S}(\frac{1}{N}\sum_{n=1}^{N}|x^s_n-\hat{x}^{s'}_n|)}{S^2}\quad \text{(Lemma \ref{lemma1})}.
\label{ineq1}
\end{align}
Note that $x^s_n-x^{truth}_n \sim N(0,\sigma^2_s)$, $\hat{x}^{s'}_n-x^{truth}_n \sim N(0,\sigma^2_{s'}+\delta^2_{s'})$ where $\sigma^2_s$ is the $s$-th user's error variance, $\sigma^2_{s'}$ and $\delta^2_{s'}$ are the $s'$-th user's error and noise variance, and $x^{truth}_n$ represents the true value of the $n$-th object. Then $x^s_n-\hat{x}^{s'}_n \sim N(0,\sigma^2_s+\sigma^2_{s'}+\delta^2_{s'})$, as $x^s_n-\hat{x}^{s'}_n = x^s_n-x^{truth} + x^{truth} -\hat{x}^{s'}_n$. Denote $\Sigma^2_{s,s'} = \sigma^2_s+\sigma^2_{s'}+\delta^2_{s'}$. We have: 

\begin{align}
\mathbb{E}(|x^s_n - \hat{x}^{s'}_n|) & = \int_{\mathbb{R}}|x|\frac{1}{\sqrt{2\pi}\Sigma_{s,s'}}\exp\{-\frac{x^2}
{2\Sigma^2_{s,s'}}\}dx\notag\\
& =2 \int_{0}^{\infty}|x|\frac{1}{\sqrt{2\pi}\Sigma_{s,s'}}\exp\{-\frac{x^2}{2\Sigma^2_{s,s'}}\}dx\notag\\
& =\sqrt{\frac{2}{\pi}}\sqrt{\sigma^2_s+\sigma^2_{s'}+\delta^2_{s'}}.
\end{align}

Based on this and strong law of large numbers, we have an almost sure convergence estimator: 
\begin{align}
\frac{1}{N}\sum_{n=1}^{N}|x^s_n-\hat{x}^{s'}_n| = \sqrt{\frac{2}{\pi}}\sqrt{\sigma^2_s+\sigma^2_{s'}+\delta^2_{s'}}.
\end{align}
By substituting it into Eq.\ (\ref{ineq1}), we have:
\begin{align}
&\frac{1}{N}\sum_{n=1}^{N}|x^*_n-\hat{x}^*_n|
\leq \sqrt{\frac{2}{\pi}}\frac{1}{S^2}\sum_{s=1}^{S}\sum_{s'=1}^{S}\sqrt{\sigma^2_s+\sigma^2_{s'}+\delta^2_{s'}}.
\label{new_start}
\end{align}

Denote $Y_{s,s'} = \sqrt{\sigma^2_s+\sigma^2_{s'}+\delta^2_{s'}}$, and $Y_{s,s'}$ is i.i.d. To simplify the notation, we use $Y$ to denote $Y_{s,s'}$ in the following. Accordingly, the p.d.f.\ of $Y$ is $h(y) = 2\frac{\lambda^2_1\lambda_2}{\lambda_2-\lambda_1}y^3e^{-\lambda_1y^2}-2\frac{\lambda^2_1\lambda_2}
{(\lambda_2-\lambda_1)^2}(ye^{-\lambda_1y^2}-ye^{-\lambda_2y^2})$. Thus, $\mathbb{E}(Y)= \sqrt{\pi}(\frac{3\lambda_2}{4\sqrt{\lambda_1}(\lambda_2-\lambda_1)}+\frac{\lambda_1^2-\lambda_2\sqrt{\lambda_1\lambda_2}}
{\sqrt{2\lambda_2}(\lambda_2-\lambda_1)^2})$, and $\mathbb{E}(Y^2) = \frac{2\lambda_2+\lambda_1}{\lambda_1\lambda_2}$.
Based on Eq.\ (\ref{new_start}), we have:
\begin{align}
&\frac{1}{N}\sum_{n=1}^{N}|x^*_n-\hat{x}^*_n|\notag \notag\leq \sqrt{\frac{2}{\pi}}\frac{1}{S^2}\sum_{s,s'\leq S}\sqrt{\sigma^2_s+\sigma^2_{s'}+\delta^2_{s'}}\\
&\leq \sqrt{\frac{2}{\pi}}|\frac{1}{S^2}\sum_{s,s'\leq S}\sqrt{\sigma^2_s+\sigma^2_{s'}+\delta^2_{s'}}\!\! -\!\! \mathbb{E}(Y)|+ \sqrt{\frac{2}{\pi}}\mathbb{E}(Y).
\label{ineq2}
\end{align}
Based on Eq.\ (\ref{ineq2}), we have:
\begin{align}
&\Pr\{\frac{1}{N}\sum_{n=1}^{N}|x^*_n-\hat{x}^*_n| \geq \alpha\}\notag\\
\leq& \Pr\{\sqrt{\frac{2}{\pi}}|\frac{1}{S^2}\sum_{s\leq S}\sum_{s'\leq S}\sqrt{\sigma^2_s+\sigma^2_{s'}+\delta^2_{s'}} - \mathbb{E}(Y)|\geq\frac{\alpha}{2}\}\notag\\
&+\Pr\{\sqrt{\frac{2}{\pi}}\mathbb{E}(Y)\geq\frac{\alpha}{2}\}\quad\quad \text{(Chebyshev's inequality)}\notag \\
\leq&\sqrt{\frac{2}{\pi}}\frac{\text{Var}(\frac{1}{S^2}\sum_{s}\sum_{s'}Y_{s,s'})}{(\alpha/2)^2}+ \Pr\{\sqrt{\frac{2}{\pi}}\mathbb{E}(Y)\geq\frac{\alpha}{2}\} \notag\\
=&4\sqrt{\frac{2}{\pi}}\frac{\frac{1}{S^2}\text{Var}(Y)}{(\alpha/2)^2} + \Pr\{\sqrt{\frac{2}{\pi}}\mathbb{E}(Y)\geq\frac{\alpha}{2}\} \leq \beta.
\label{assp1}
\end{align}
Once the exponential distributions are given, the probability in Eq.\ (\ref{assp1}) is either $0$ or $1$. Note that the smaller $\beta$, the better utility we can obtain. We can achieve $\Pr\{\sqrt{\frac{2}{\pi}}\mathbb{E}(Y)\geq\frac{\alpha}{2}\}=0$ by assuming that $\alpha > \frac{2\sqrt{2}}{\sqrt{\pi}}\mathbb{E}(Y)$.
Thus, Eq.\ (\ref{assp1}) can be reduced to $\text{Var}(Y) \leq \frac{\sqrt{\pi}\alpha^2\beta S^2}{4\sqrt{2}}$.
Moreover,
\begin{eqnarray}
\mathbb{E}(Y^2) \leq \frac{\sqrt{\pi}\alpha^2\beta S^2}{4\sqrt{2}} + (\mathbb{E}(Y))^2\leq \frac{\sqrt{\pi}\alpha^2\beta S}{4\sqrt{2}} + (\frac{\alpha\sqrt{\pi}}{2\sqrt{2}} + \sqrt{2})^2. \notag
\end{eqnarray}
Since $\mathbb{E}(Y^2) = \frac{2\lambda_2+\lambda_1}{\lambda_1\lambda_2}$, we have:
\begin{eqnarray}
\frac{2\lambda_2+\lambda_1}{\lambda_1\lambda_2}&\leq& \frac{\sqrt{\pi}\alpha^2\beta S^2}{4\sqrt{2}} + (\frac{\alpha\sqrt{\pi}}{2\sqrt{2}} + \sqrt{2})^2.
\end{eqnarray}
By substituting $\frac{1}{\lambda_2} = c\frac{1}{\lambda_1}$, we can obtain an upper bound for $c$ to obtain $(\alpha,\beta)$-utility:
\begin{eqnarray}
c\leq \lambda_1\sqrt{\pi}(\frac{\alpha^2\beta S^2}{4\sqrt{2}}+\frac{\alpha^2\sqrt{\pi}}{8}+\alpha+\frac{2}{\sqrt{\pi}})-2 \triangleq C_{\lambda_1,\alpha,\beta,S}.
\label{upperbound}
\end{eqnarray}
As $\sqrt{\frac{2}{\pi}}\mathbb{E}(Y)<\frac{\alpha}{2}$ , we can also obtain a lower bound for $\alpha$, namely $\alpha_{\lambda,c} = \frac{2\sqrt{2}}{\sqrt{\lambda_1}(1-c)}(\frac{3}{4}-\frac{c(c+\sqrt{c}+1)}{\sqrt{2}(1+\sqrt{c})})$, which completes the proof of our theorem.
\end{proof}
This theorem reveals the relationship between the noise and the utility for the proposed mechanism. The upper bound of $c$ specifies the noise level that the proposed mechanism can afford to achieve $(\alpha,\beta)$-utility. From the equation that defines the upper bound of $c$, we can observe the following: (1) When $\alpha$ and $\beta$ become smaller or larger, the upper bound of $c$ decreases or increases, which indicates that better utility requires smaller noise and vice versa. (2) The upper bound of $c$ increases with the increase in the number of users $S$. This means that we can tolerate more noise when more users contribute their information to the aggregation tasks. (3) As $\lambda_1$ captures error distributions in original data, a larger $\lambda_1$ indicates better information quality and correspondingly the mechanism can tolerate more noise.

\subsection{Privacy Analysis}
In this section, we analyze how noise level $c$ is related to user privacy. The traditional differential privacy definition provides the protection of user privacy against information leakage through statistical query results, in which a trusted server is assumed and thus it does not fit the privacy-preserving truth discovery scenario. Recently, local differential privacy \cite{duchi2013local,dwork2013algorithmic,kairouz2014extremal} is proposed to deal with the scenario where individual users do not trust the server. Based on local differential privacy, we adopt the following privacy definition to quantify the user privacy:

\begin{definition}[($\epsilon$,$\delta$)-Local Differential Privacy]
We say a mechanism $\mathcal{M}$ satisfies ($\epsilon$,$\delta$)-Local Differential Privacy, if for any subset $\mathbb{S}\subseteq\mathbb{R}$ and two different records $x^1$ and $x^2$, the following inequality holds:
\begin{align}
\Pr\{\mathcal{M}(x^1) \in \mathbb{S}\}\leq e^{\epsilon} \Pr\{\mathcal{M}(x^2) \in \mathbb{S}\} + \delta.
\end{align}
\label{our_privacy_definition}
\end{definition}

This definition compares the probability of observing the perturbed value of two different records $x^1$ and $x^2$ in the same range. With two distinguishable pieces of information $x^1$ and $x^2$, an ideal perturbation mechanism should perturb them to indistinguishable values to preserve user privacy in crowdsourced data collection. As can be seen, this definition is stronger than traditional differential privacy which compares the probability of observing similar query outputs on two different databases with one record difference.

Next, we define sensitive information for each user and derive its relationship to the hyper-parameter $\lambda_1$ which controls the error variance distribution.
\begin{definition}[Sensitive Information]
The sensitive information of the $s$-th user is denoted by 
\begin{eqnarray}
\Delta_s = \max_{x^1_s,x^2_s\in D}|x^1_s - x^2_s|,
\end{eqnarray}
where $x^1_s$ and $x^2_s$ are two entries claimed by the $s$-th user about the same object.
\label{si}
\end{definition}
The sensitive information, $\Delta_s$, measures the range of information claimed by the $s$-th user. Intuitively, $\Delta_s$ is related to $\lambda_1$, as $\lambda_1$ controls the variance of users' error and large variance (small $\lambda_1$) leads to large range of values $\Delta_s$. The following lemma formally defines their relationship.
\begin{lemma}
The p.d.f. of the errors' variance is $f(z)=\lambda_1 e^{-\lambda_1 z}$. The sensitive information about the $s$-th user, $\Delta_s$, satisfies that $\Delta_s =|x^1_s-x^2_s| \leq\frac{\gamma_s}{\lambda_1}$ with probability at least $\eta(1-\frac{2e^{-b^2/2}}{b})$, where $\gamma_s=b\sqrt{2\ln{\frac{1}{1-\eta}}}$, $\eta$ and $b$ are constants.
\label{delta}
\end{lemma}
\begin{proof}
	As the $s$-th user's error follows $N(0,\sigma_s^2)$, the $s$-th user's information $x_s\sim N(x^{truth},\sigma_s^2)$ where $x^{truth}$ is the true value and $\sigma_s^2$ is the variance drawn from the exponential distribution with $\lambda_1$. Based on the property of light tail of exponential distribution, given a sufficient large number $M$, $\Pr\{\sigma\leq M\}=1-e^{-\lambda_1M^2} = \eta$, which implies $M=\frac{\sqrt{\ln{\frac{1}{1-\eta}}}}{\sqrt{\lambda_1}}$. As $\lambda_1$ becomes bigger, $M$ could be smaller; vice versa. Based on the assumption that most of the users are reliable, $\lambda_1$ should be larger than $1$. Consequently, $M\leq\frac{\sqrt{\ln{\frac{1}{1-\eta}}}}{\lambda_1}$.
	
	Now, we try to bound $\Delta_s$. Let $x^1_s$ and $x^2_s$ be two pieces of information claimed by the $s$-th user about the same object. Thus $x^1_s-x^2_s\sim N(0,2\sigma^2)$. Based on Gaussian Tail Inequality, we have $
	\Pr\{|x^1_s-x^2_s|>b\sqrt{2}\sigma\}\leq \frac{2e^{-b^2/2}}{b}$,
	which implies $\Delta_s =|x^1_s-x^2_s|\leq b\sqrt{2}\sigma$ with probability at least $1-\frac{2e^{-b^2/2}}{b}$. Let $\gamma_s=b\sqrt{2\ln{\frac{1}{1-\eta}}}$. Then we have $\Delta_s =|x^1-x^2|\leq b\sqrt{2}\sigma\leq\frac{b\sqrt{2\ln{\frac{1}{1-\eta}}}}{\lambda_1}=\frac{\gamma_s}{\lambda_1}$ with probability at least $\eta(1-\frac{2e^{-b^2/2}}{b})$.
\end{proof}

From this lemma, we can see that the sensitive information of each user is inversely proportional to $\lambda_1$, which measures the quality of users. The bigger $\lambda_1$ is, the smaller the error's variance is and the smaller sensitive information of the user. In the following discussion, we choose that $\Delta_s =|x^1_s-x^2_s|=\frac{b\sqrt{2\ln{\frac{1}{1-\eta}}}}{\lambda_1}$.

Next, we prove the main result of privacy analysis which links noise level to local differential privacy under the proposed mechanism.
\begin{theorem}
Consider a perturbation mechanism $\mathcal{M}$ with parameter $\lambda_2$, where $1/\lambda_2=c/\lambda_1$. Based on Definition \ref{si} , $\mathcal{M}$ satisfies ($\epsilon$,$\delta$)-Local Differential Privacy in terms of the $s$-th user, provided $c\geq\frac{\gamma_s^2}{2\lambda_1\epsilon\ln(\frac{1}{1-\delta})}$, where $\gamma_s=b\sqrt{2\ln{\frac{1}{1-\eta}}}$.\label{privacy}
\end{theorem}
\begin{proof}
Based on the mechanism, the $s$-th user draws his noise variance from an exponential distribution with parameter $\lambda_2$. Assume that the noise variance is $y$, we have:
\begin{align}
&\Pr\{\mathcal{M}(x^1)=x\} = \frac{1}{\sqrt{2\pi y}} \exp(-\frac{(x-x^1)^2}{2y})\notag\\
\leq & \frac{1}{\sqrt{2\pi y}} \exp(-\frac{(x-x^2)^2-(x^2-x^1)^2}{2y})\notag\\
= & \exp(\frac{(x^2-x^1)^2}{2y}) \frac{1}{\sqrt{2\pi y}} \exp(-\frac{(x-x^2)^2}{2y})\notag\\
\leq & \exp(\frac{\Delta_s^2}{2y}) \Pr\{\mathcal{M}(x^2)=x\} \leq  e^{\epsilon} \Pr\{\mathcal{M}(x^2)=x\}.
\label{approach1}
\end{align}
Obviously, $\Pr\{\mathcal{M}(x^1)=x\}\leq e^{\epsilon}\Pr\{\mathcal{M}(x^2)=x\}$, if and only if $y\geq \frac{\Delta_s^2}{2\epsilon}$. Since $y$ follows exponential distribution with parameter $\lambda_2$, we constrain that the event $\{y:y\geq\frac{\Delta_s^2}{2\epsilon}\}$ happens with at least $1 - \delta$ probability. Namely, $\Pr(\{y:y\geq\frac{\Delta_s^2}{2\epsilon}\}) \geq 1 - \delta$, where $\delta\in[0,1]$. Thus,
$\Pr(\{y:y\geq\frac{\Delta_s^2}{2\epsilon}\})
=\exp(-\frac{\lambda_2\Delta_s^2}{2\epsilon})\geq 1-\delta$.
Then, $\frac{\lambda_2\Delta_s^2}{2\epsilon}\leq\ln(\frac{1}{1-\delta})$. Since $\frac{1}{\lambda_2}=c\frac{1}{\lambda_1}$, we have $c\geq\frac{\lambda_1\Delta_s^2}{2\epsilon\ln(\frac{1}{1-\delta})}$. Based on Lemma \ref{delta}, we have $c\geq\frac{\gamma_s^2}{2\lambda_1\epsilon\ln(\frac{1}{1-\delta})}$, where $\gamma_s=b\sqrt{2\ln{\frac{1}{1-\eta}}}$.

Note that the domain of noise variance is $\mathbb{R}^{+}$. Let us divide $\mathbb{R}^{+}$ as $\mathbb{R}^{+}=R_1\cup R_2$, where
$R_1=\{\rho^2\in\mathbb{R}^{+}:\rho^2\geq\frac{\Delta_s^2}{2\epsilon}\}$ and $R_2=\{\rho^2\in\mathbb{R}^{+}:\rho^2<\frac{\Delta_s^2}{2\epsilon}\}$. Denote $\mathcal{M}(x,\rho^2)$ as the mechanism $\mathcal{M}$ adding noise $N(0,\rho^2)$ to the record $x$. Let $\mathbb{S}\subseteq\mathbb{R}$ be given. We adopt the idea used in Gaussian mechanism \cite{dwork2013algorithmic} to build two different subsets of $\mathbb{S}$, i.e., $\mathbb{S}_1$ and $\mathbb{S}_2$, in the following way. For a specific output $\mathcal{M}(x,\rho^2)\in \mathbb{S}$, we claim $\mathcal{M}(x,\rho^2) \in \mathbb{S}_1$ if $\rho^2\in R_1$ or $\mathcal{M}(x,\rho^2)\in \mathbb{S}_2$ if $\rho^2\in R_2$. Therefore, the probability of event $\mathcal{M}(x,\rho^2)$ belonging to $\mathbb{S}_1$ equals to that of event $\rho^2$ belonging to $R_1$. Similar relation holds between the event $\mathcal{M}(x,\rho^2)\in \mathbb{S}_2$ and the event $\rho^2\in R_2$.

Thus, we have
\begin{align}
\Pr_{\rho^2\in\mathbb{R}^{+}}\{\mathcal{M}(x^1,\rho^2) \in \mathbb{S}\} &\leq (\Pr_{\rho^2\in R_1} + \Pr_{\rho^2\in R_2})\{\mathcal{M}(x^1,\rho^2) \in \mathbb{S}_2\}\notag\\
&\leq \Pr_{\rho^2\in R_1}\{\mathcal{M}(x^1,\rho^2) \in \mathbb{S}_1\} + \delta \notag\\
&\leq e^{\epsilon}(\Pr_{\rho^2\in\mathbb{R}^{+}}\{\mathcal{M}(x^2,\rho^2) \in \mathbb{S}\})+\delta \notag
\end{align}
yielding $(\epsilon,\delta)$-local differential privacy for the proposed perturbation mechanism $\mathcal{M}$.
\end{proof}

From Theorem \ref{privacy}, we can conclude that to achieve stronger privacy, namely for smaller $\epsilon$, the noise level $c$ has to be greater than a certain threshold. This is consistent with intuitions that more noise leads to stronger privacy protection. The lower bound of $c$ is related to $\lambda_1$ and privacy parameters $\epsilon$ and $\delta$.  Smaller $\epsilon$ and $\delta$ (stronger privacy protection) ask for a bigger bound for the noise level.  The bigger $\lambda_1$, the smaller the variance of users' error, and thus less noise is required to guarantee privacy. Since the mechanism of generating perturbed data is the same across users, Theorem \ref{privacy} is applicable to each individual user.

\subsection{Utility-Privacy Trade-off}

Based on Theorem \ref{utility} (utility) and Theorem \ref{privacy} (privacy), we can now analyze the trade-off between utility and privacy as shown in the following theorem:

\begin{theorem}(Utility Privacy Trade-off)
Consider a truth discovery algorithm $\mathcal{A}$ with perturbation mechanism $\mathcal{M}$ and an input data set $D$. Based on Assumption \ref{assumption1} and Definition \ref{si}, $\forall \alpha > \frac{2\sqrt{2}}{\sqrt{\lambda_1}(1-c)}(\frac{3}{4}-\frac{c(c+\sqrt{c}+1)}{\sqrt{2}(1+\sqrt{c})})$,
the algorithm $\mathcal{A}$ with perturbation mechanism $\mathcal{M}$ satisfies $(\alpha,\beta)$-Utility and $(\epsilon,\delta)$-Local Differential Privacy, provided $c\leq  \lambda_1\sqrt{\pi}(\frac{\alpha^2\beta S^2}{4\sqrt{2}}+\frac{\alpha^2\sqrt{\pi}}{8}+\alpha+\frac{2}{\sqrt{\pi}})-2$ and $c\geq\frac{\gamma_s^2}{2\lambda_1\epsilon\ln(\frac{1}{1-\delta})}$, where $\gamma_s=b\sqrt{2\ln{\frac{1}{1-\eta}}}$.
\label{tradeoff_theorem}
\end{theorem}
\begin{proof}
Based on the Theorem \ref{utility} and \ref{privacy}, Theorem \ref{tradeoff_theorem} holds immediately.
\end{proof}

Theorem \ref{tradeoff_theorem} provides a guideline on how to choose a proper $c$ to achieve the trade-off between utility and privacy. To have a valid $c$, we must have the upper bound of $c$ derived from utility analysis to be greater than or equal to the lower bound of $c$ derived from privacy analysis. Especially, to have at least one $c$ exist, the two bounds should be the same, and thus we have:
 \begin{equation}
 \lambda_1\sqrt{\pi}(\frac{\alpha^2\beta S^2}{4\sqrt{2}}+\frac{\alpha^2\sqrt{\pi}}{8}+\alpha+\frac{2}{\sqrt{\pi}})-2=\frac{\gamma_s^2}{2\lambda_1\epsilon\ln(\frac{1}{1-\delta})}.
 \end{equation}
 It is obvious that stronger privacy (smaller $\epsilon$ and $\delta$) can be satisfied when sacrificing utility (increase in $\alpha$ and $\beta$), and better utility can be achieved when privacy is comprised. This trade-off is related to the characteristics of data, i.e., the error distribution in the original data, which is controlled by hyper-parameter $\lambda_1$. A larger $\lambda_1$ indicates a higher chance that users' original information is similar enough, and thus strong privacy and good utility are possible. Similarly, a smaller $\lambda_1$ leads to more challenging privacy protection in which less privacy and utility gain are expected. In the following section, we will experimentally verify this trade-off.

\section{Experiment}
\label{sec:exp}

In the previous section, we quantified the trade-off between the utility of aggregated results and the privacy of users. Now, we illustrate this result via a set of experiments: (1) We demonstrate how the proposed mechanism achieves good privacy and utility on simulated datasets, and evaluate the performance under different scenarios. (2) The privacy and utility trade-off is further demonstrated on a real crowd sensing system. We also demonstrate how good utility is achieved by the proposed mechanism which can automatically assign user weights based on information quality. (3) We show that the proposed method is scalable to large-scale data by conducting efficiency tests.

\subsection{Experiments on Synthetic Dataset}

In this part, we show experimental results on synthetic datasets. As mentioned in Section \ref{sec:method}, the error made by each individual user can be captured by a normal distribution $N(0,{\sigma_s}^2)$, where the variance indicates the quality of his information. Therefore, we simulate 150 users with various qualities by setting different ${\sigma_s}^2$, and generate their provided information for 30 objects based on both the ground truth information and the sampled error. We regard this dataset as the original data contributed by the users.

For each individual user, we follow Algorithm 2 to perturb his information. Specifically, we choose a hyper-parameter $\lambda_2$, and generate each user's noise parameter ${\delta_s}^2$ from exponential function with $\lambda_2$. Then each user's data is perturbed by injecting sampled random noise based on Gaussian distribution with ${\delta_s}^2$ as variance. We then conduct truth discovery on the perturbed data.

To measure the utility of aggregation, we compare the aggregated results based on original data and perturbed data, and quantify their difference. Here, we adopt the commonly used $L^1$-norm distance, i.e., the mean of absolute distance (MAE) on all objects. For this measure, lower value indicates better utility.

Note that the privacy parameter $\epsilon$ is defined in a different way compared with traditional differential privacy. As shown in Definition \ref{our_privacy_definition}, the same $\epsilon$ in local differential privacy indicates stronger privacy protection as the definition is based on the perturbation of one record. Nonetheless, we can still observe low $\epsilon$ in the following experiments.

\smallskip \noindent \textsf{\textbf{Utility-Privacy Trade-off}}. Figure \ref{trade_syn} plots the utility-privacy trade-off on the synthetic dataset. Figure \ref{1_1} shows the trade-off in terms of privacy parameters $\epsilon$ and MAE. To provide some intuition about how much noise the approach can tolerate, we show the average noise level corresponding to different $\epsilon$ in Figure \ref{1_3}. We can see that in order to guarantee stronger privacy (smaller $\epsilon$), larger noise is needed. However, the added large noise only incurs small loss in utility. From Figure \ref{1_1}, we can observe that the utility changes very slowly and the magnitude is quite small compared with the added noise.  To facilitate comparison, we use the same $x$-axis and make the scale of $y$-axis the same for these plots. As can be seen, when the average of added noise reaches closer to 1, the average loss in utility is less than 0.1 (only 1/10 of the noise). This  demonstrates the advantage of the proposed mechanism in maintaining good utility even when high noise is added. 
\begin{figure}[tbh]
\centering
\subfloat [MAE]{
\includegraphics[width=0.21\textwidth]{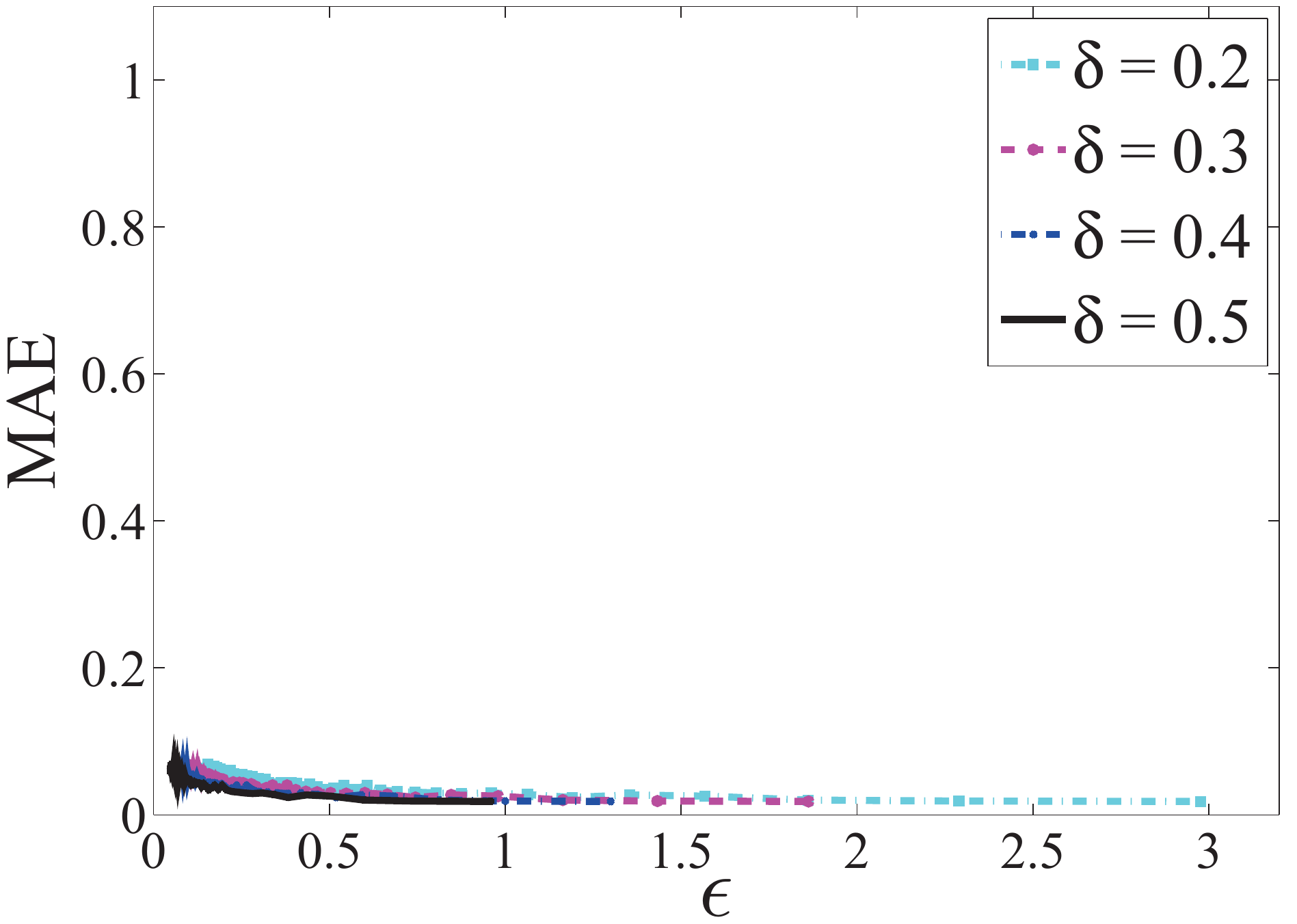}
\label{1_1}
}
\hfill
\subfloat [Noise]{
\includegraphics[width=0.21\textwidth]{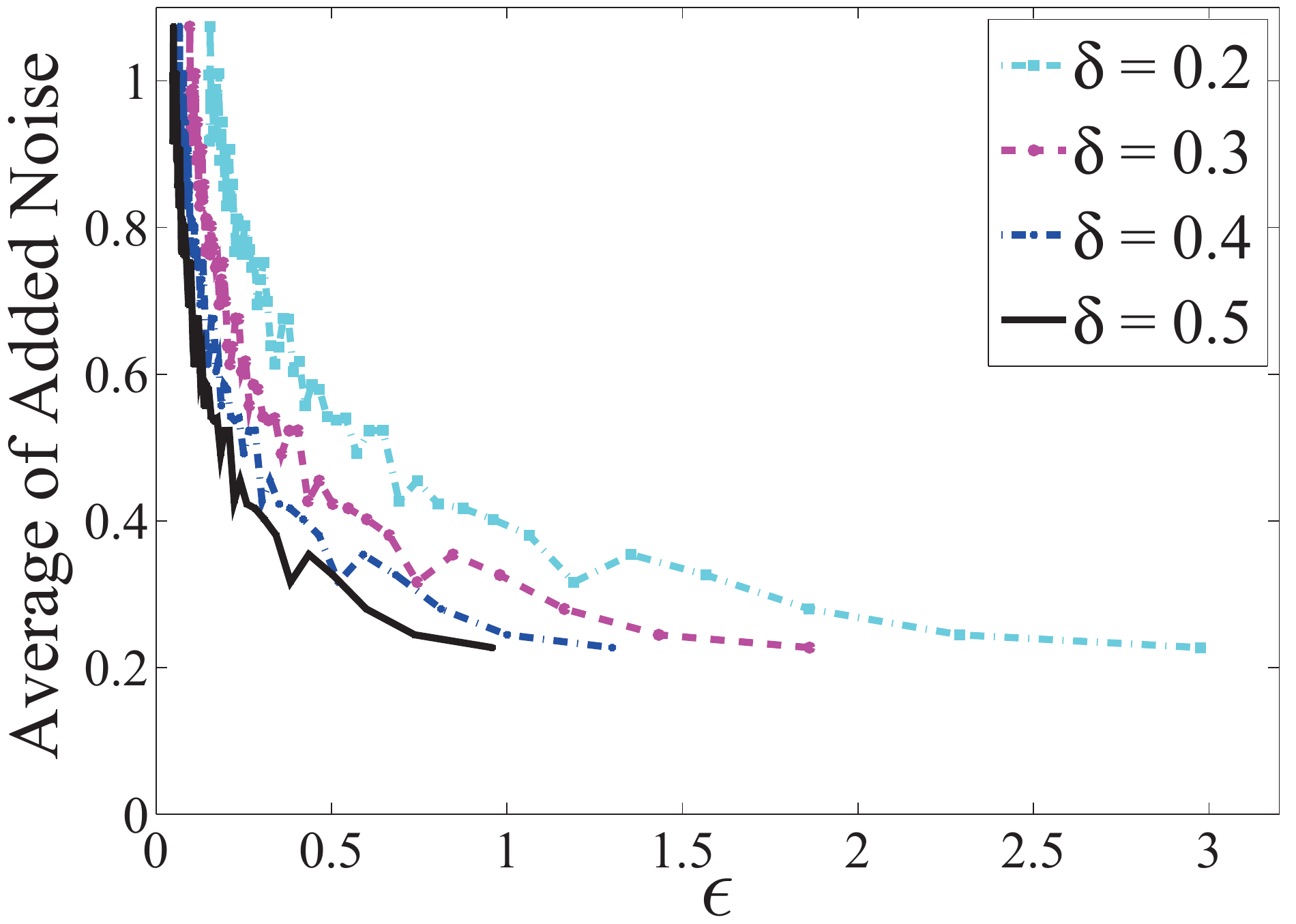}
\label{1_3}
}
\vspace{-0.1in}
\caption{Utility-Privacy Trade-off on Synthetic Dataset}
\label{trade_syn}
\end{figure}

\smallskip \noindent \textsf{\textbf{Effect of $\lambda_1$}}. In Theorem \ref{tradeoff_theorem}, we show that $\lambda_1$ is related to both utility and privacy and here we demonstrate its effect empirically.  $\lambda_1$ captures the information quality distribution of original data. As shown in Figure \ref{2_1}, when $\lambda_1$ is big, the variance of the distribution used to sample noise variance is small, so the quality of all users are relatively good. Then users tend to contribute similar information for the same object, so small noise can hide users' information. On the other hand, when $\lambda_1$ is small, user-provided information may be quite different due to the low quality controlled by the error distribution, and thus only large noise can preserve their privacy. Figure \ref{2_2} demonstrates utility variation under different $\lambda_1$.  With small $\lambda_1$, large noise has to be added so the utility will be affected more. The message we can get is that it is easier to maintain both privacy and utility if the original data has high quality and it is more challenging when the original data is noisy.
\begin{figure}[tbh]
\centering
\subfloat [MAE]{
\includegraphics[width=0.21\textwidth]{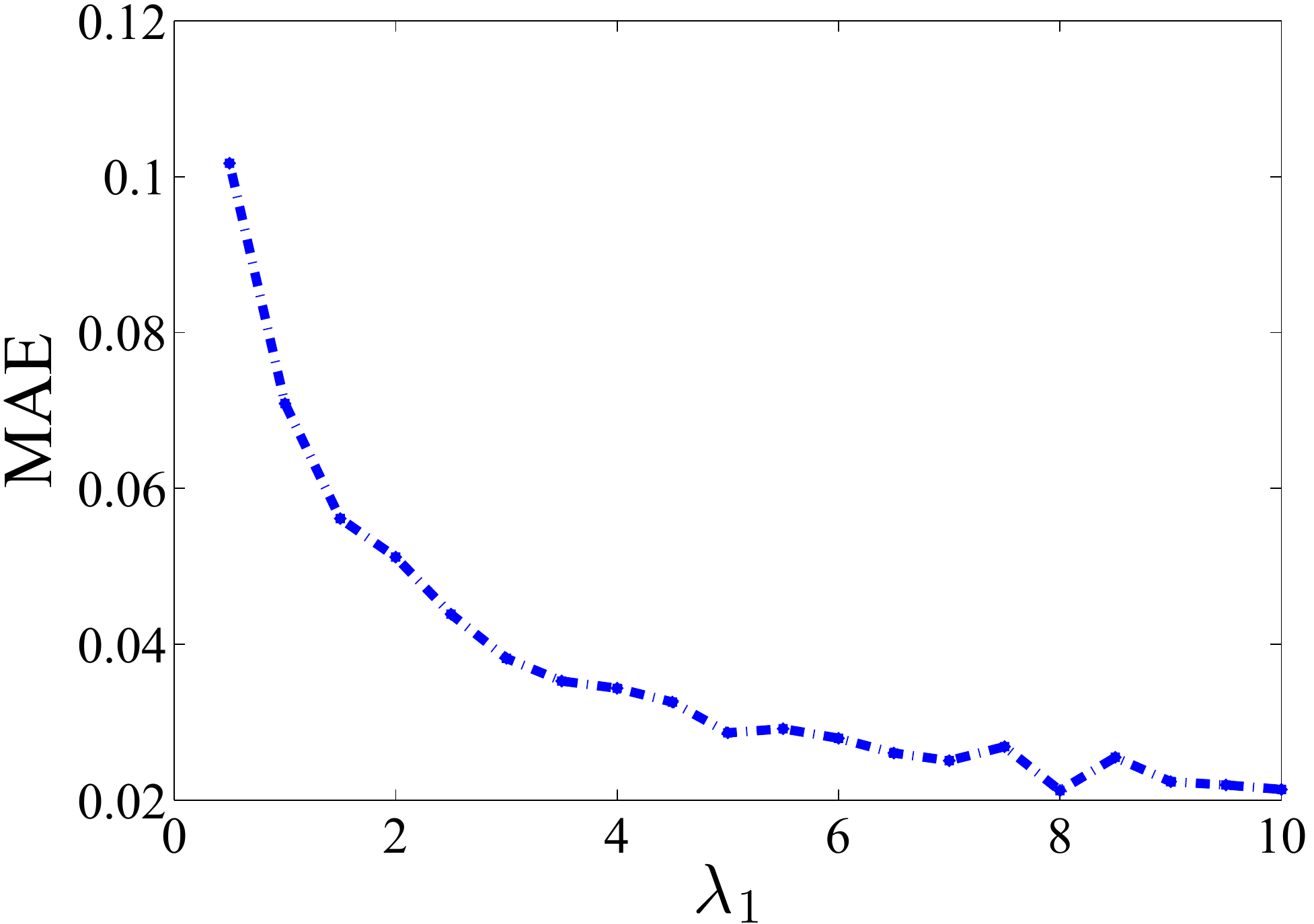}
\label{2_2}
}
\hfill
\subfloat [Noise]{
\includegraphics[width=0.21\textwidth]{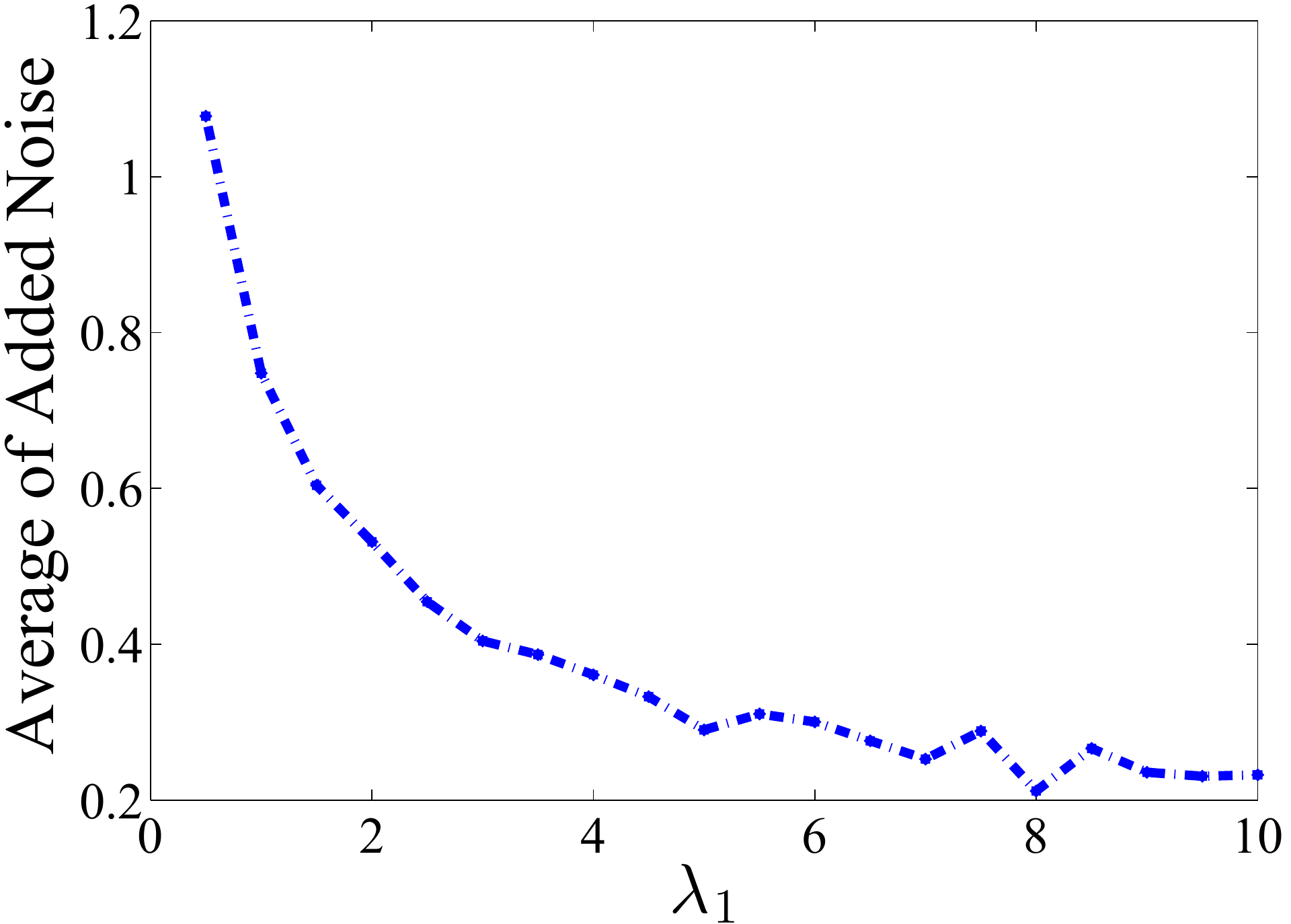}
\label{2_1}
}
\caption{Effect of $\lambda_1$ (Parameter of Error Distribution in Original Data)}
\label{effect_lambda}
\end{figure}

\smallskip \noindent \textsf{\textbf{Effect of $S$}}. Next, we study the effect of $S$, i.e., the number of users that are involved in the aggregation task. In the proposed mechanism, all users act independently to add noise and they do not rely on each other. Hence the average noise will not be affected by the number of users. This phenomenon is demonstrated in Figure \ref{3_1} in which the average of added noise keeps the same as $S$ increases. On the other hand, Figure \ref{3_2} shows that having more users can help utility. The reason is that truth discovery approaches can estimate user weights better when more information is collected, and thus obtain better aggregation results. This is consistent with our theoretical analysis on utility in Theorem \ref{utility}: To achieve the same level of utility, we can tolerate larger noise if more users are involved in the task.
\begin{figure}[tbh]
\centering
\subfloat [MAE]{
\includegraphics[width=0.21\textwidth]{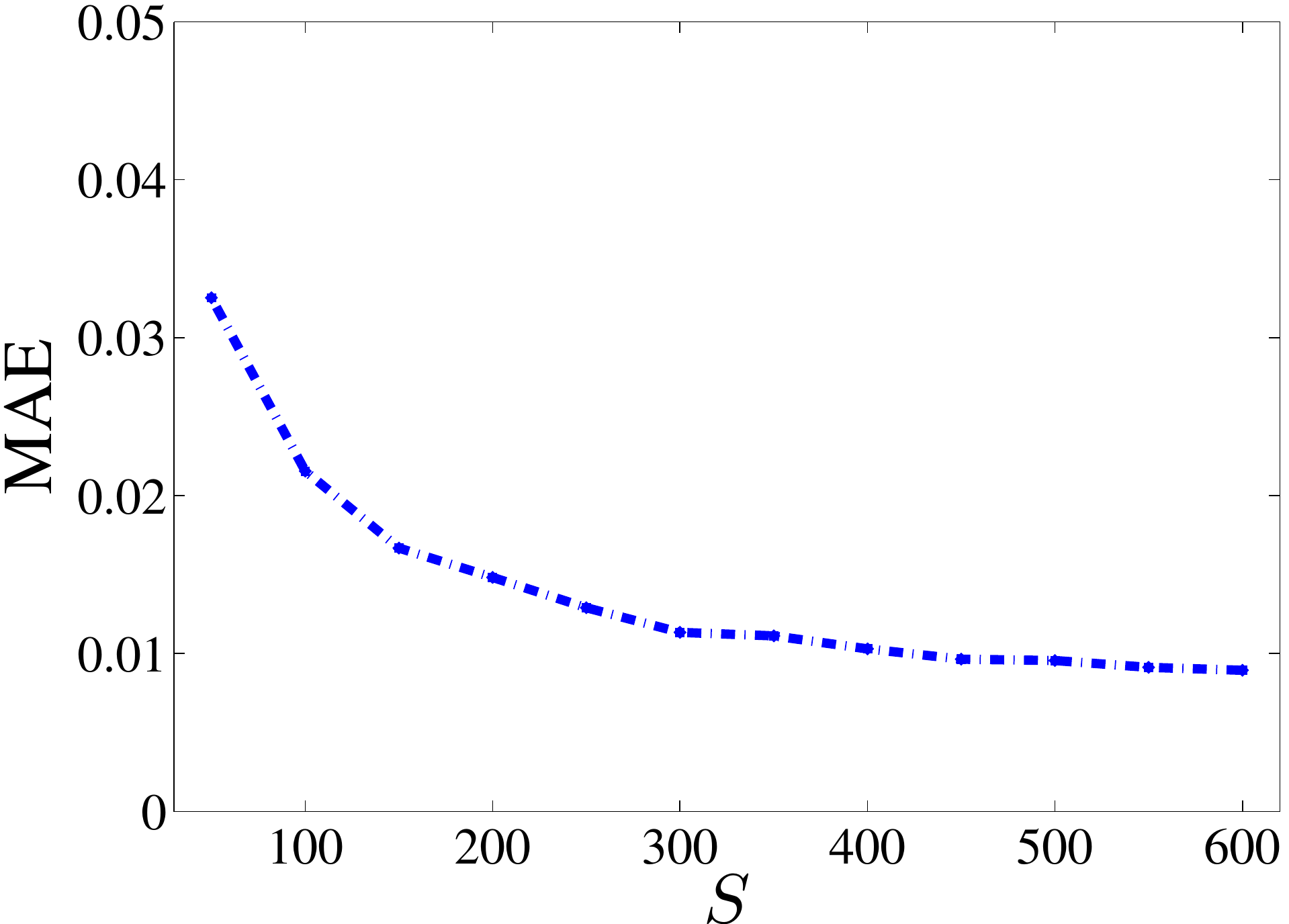}
\label{3_2}
}
\hfill
\subfloat [Noise]{
\includegraphics[width=0.21\textwidth]{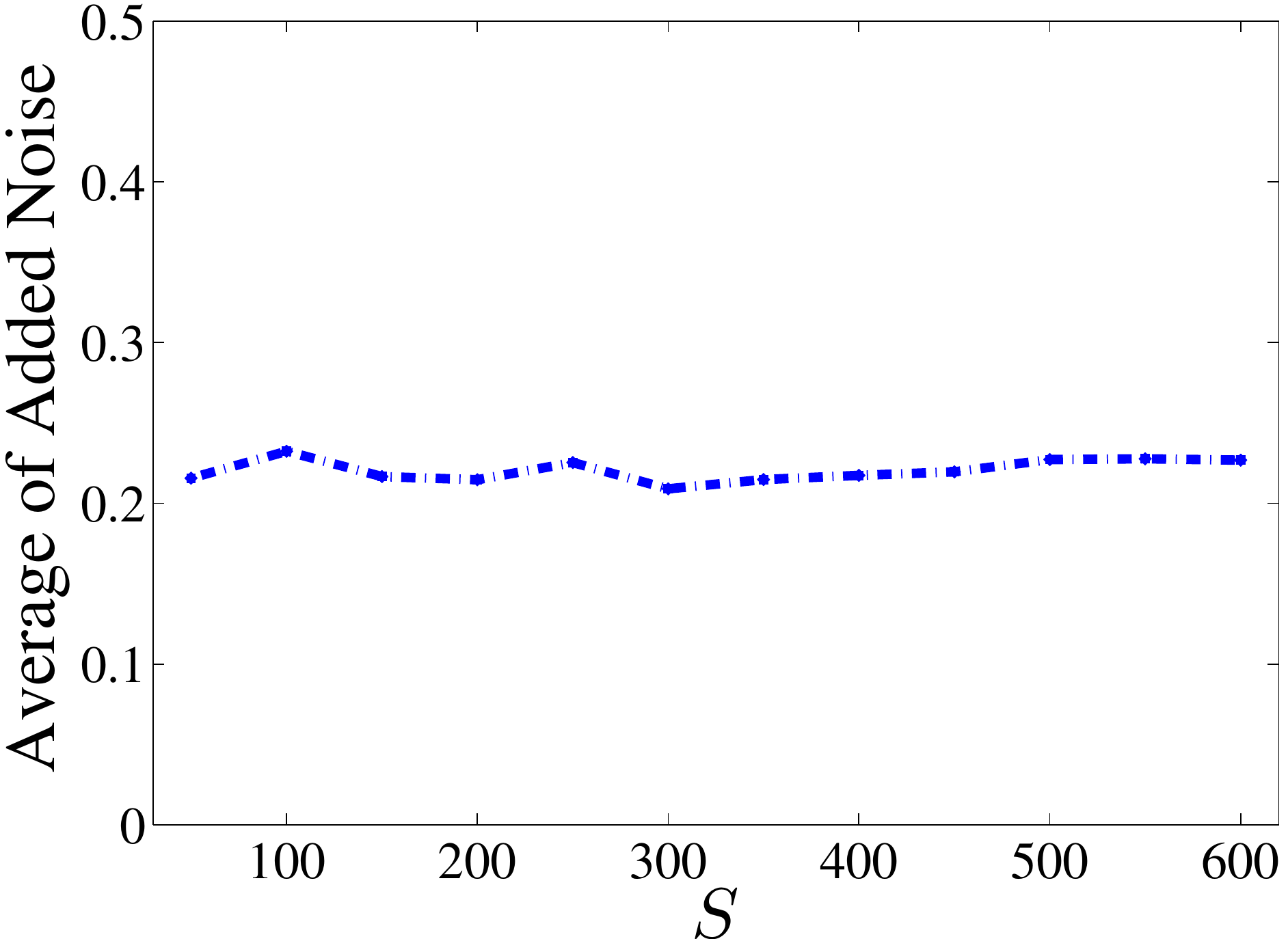}
\label{3_1}
}
\caption{Effect of $S$ (Number of Users)}
\label{effect_s}
\end{figure}

\smallskip \noindent \textsf{\textbf{Truth Discovery Methods}}.
As discussed in Algorithm 2, this mechanism can work with any specific method that satisfies the general principle of truth discovery. In the experiments shown so far, we adopt the recent CRH \cite{crh_sigmod14} method. Here we present results on a different truth discovery method to illustrate the mechanism's ability to generalize to other approaches. We apply another start-of-the-art truth discovery approach that can be applied to continuous data, GTM \cite{bo_gtm_qdb12}, and show the results in Figure \ref{trade_gtm}. The patterns of utility privacy trade-off are similar compared to those based on CRH.
\begin{figure}[tbh]
\centering
\subfloat [MAE]{
\includegraphics[width=0.21\textwidth]{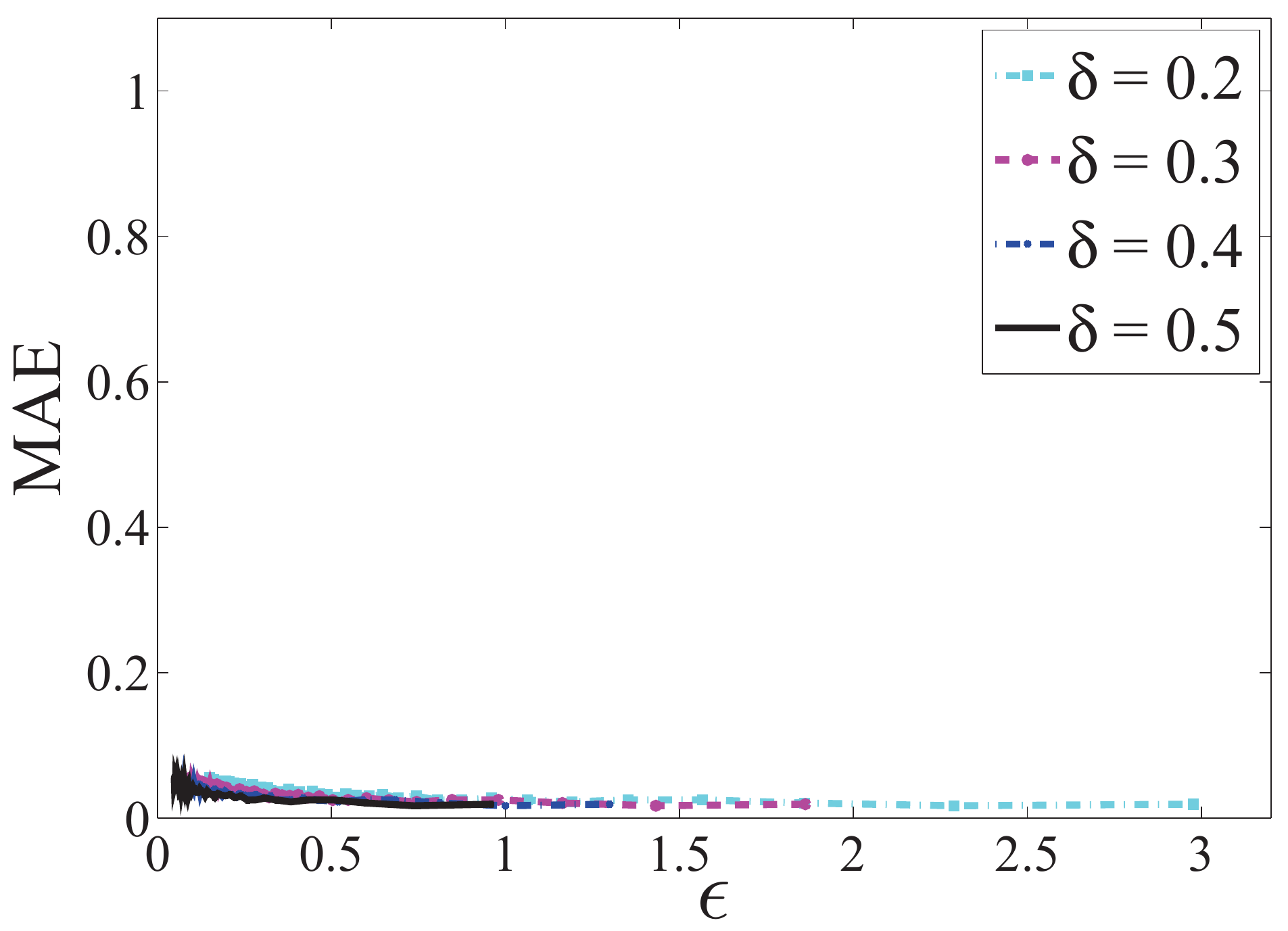}
}
\hfill
\subfloat [Noise]{
\includegraphics[width=0.21\textwidth]{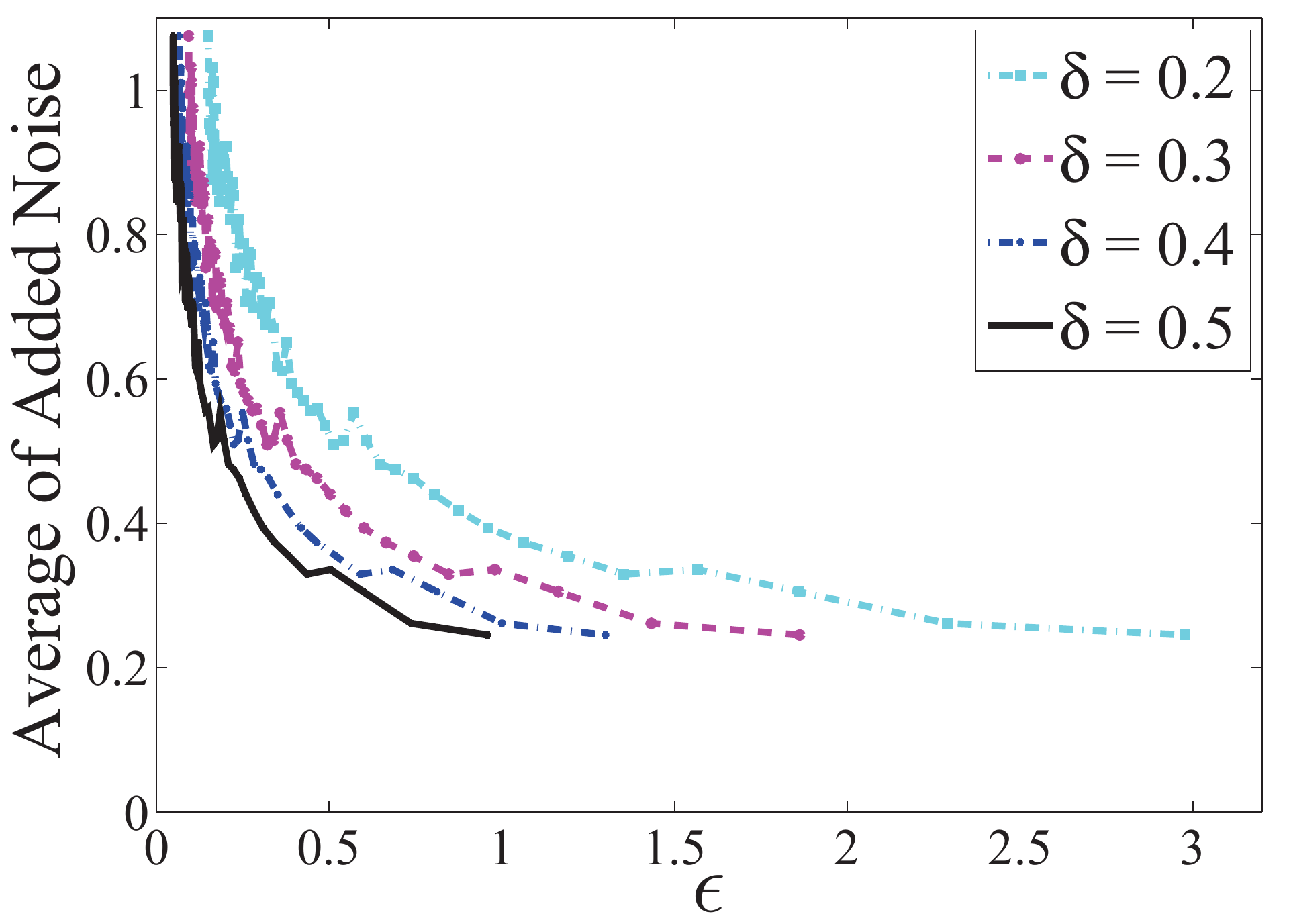}
}
\caption{Utility-Privacy Trade-off on Synthetic Dataset (GTM)}
\label{trade_gtm}
\end{figure}

\subsection{Experiments on Crowd Sensing Application}

In this section, we show experimental results on a real world crowd sensing application to illustrate the effectiveness of the proposed mechanism.  The application is indoor floorplan construction \cite{shen2013walkie,gao2014jigsaw}, which has gained growing interest as many location-based services are built upon this task. The goal is to automatically construct indoor floorplan from sensor data collected from smartphone users. However, since the private personal activities of the phone users are usually encoded in the sensor readings, the users may not willing to share their data without privacy protection guarantee. Here we focus on one task of indoor floorplan construction: Estimate the distance between two location points along a straight hallway by aggregating user data. 
Specifically, we select $129$ hallway segments as the objects and collect data from $247$ smartphone users via a developed Android app.
We then obtain the distance each user has traveled on each hallway segment by multiplying user step size by step count. Due to different walking patterns and in-phone sensor quality, the distances obtained by different users on the same segment can be quite different. The goal is to derive the true length of hallways by aggregating user-provided distance information. We let each user add noise to their original information following the procedure in Algorithm \ref{alg:proposed}. We vary the hyper parameter $\lambda_2$ to collect multiple sets of perturbed data, and the truth discovery method adopted here is still CRH.

\smallskip \noindent \textsf{\textbf{Utility-Privacy Trade-off}}. We still adopt MAE to measure aggregation utility. Figure \ref{trade_indoor} shows the utility and privacy trade-off. Compared with Figure \ref{trade_syn}, we observe the same pattern that is shown on synthetic data. This confirms that even when the added noise is quite large (strong privacy), good utility can be achieved under the proposed mechanism.
\begin{figure}[tbh]
\centering
\subfloat [MAE]{
\includegraphics[width=0.21\textwidth]{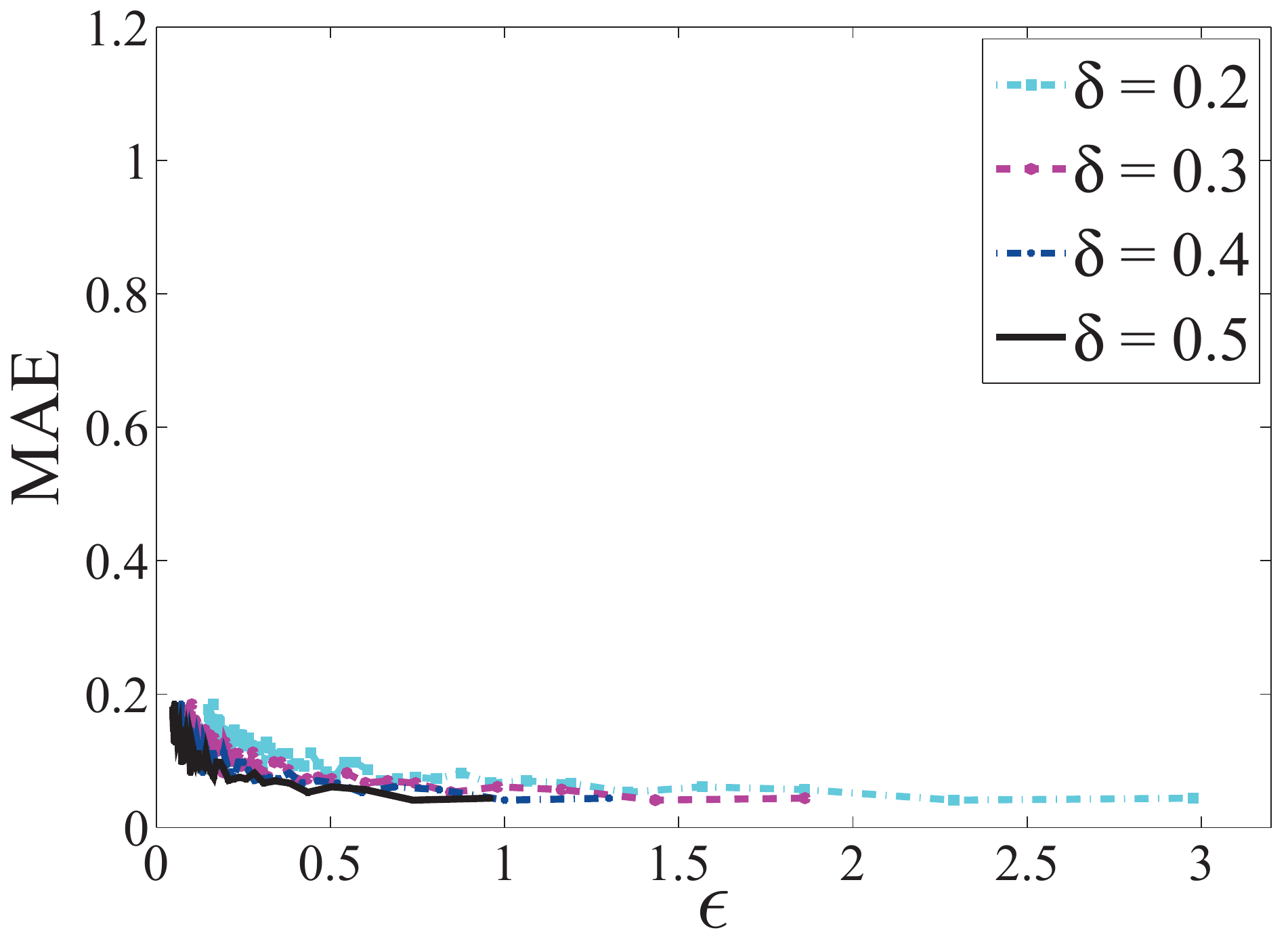}
}
\hfill
\subfloat [Noise]{
\includegraphics[width=0.21\textwidth]{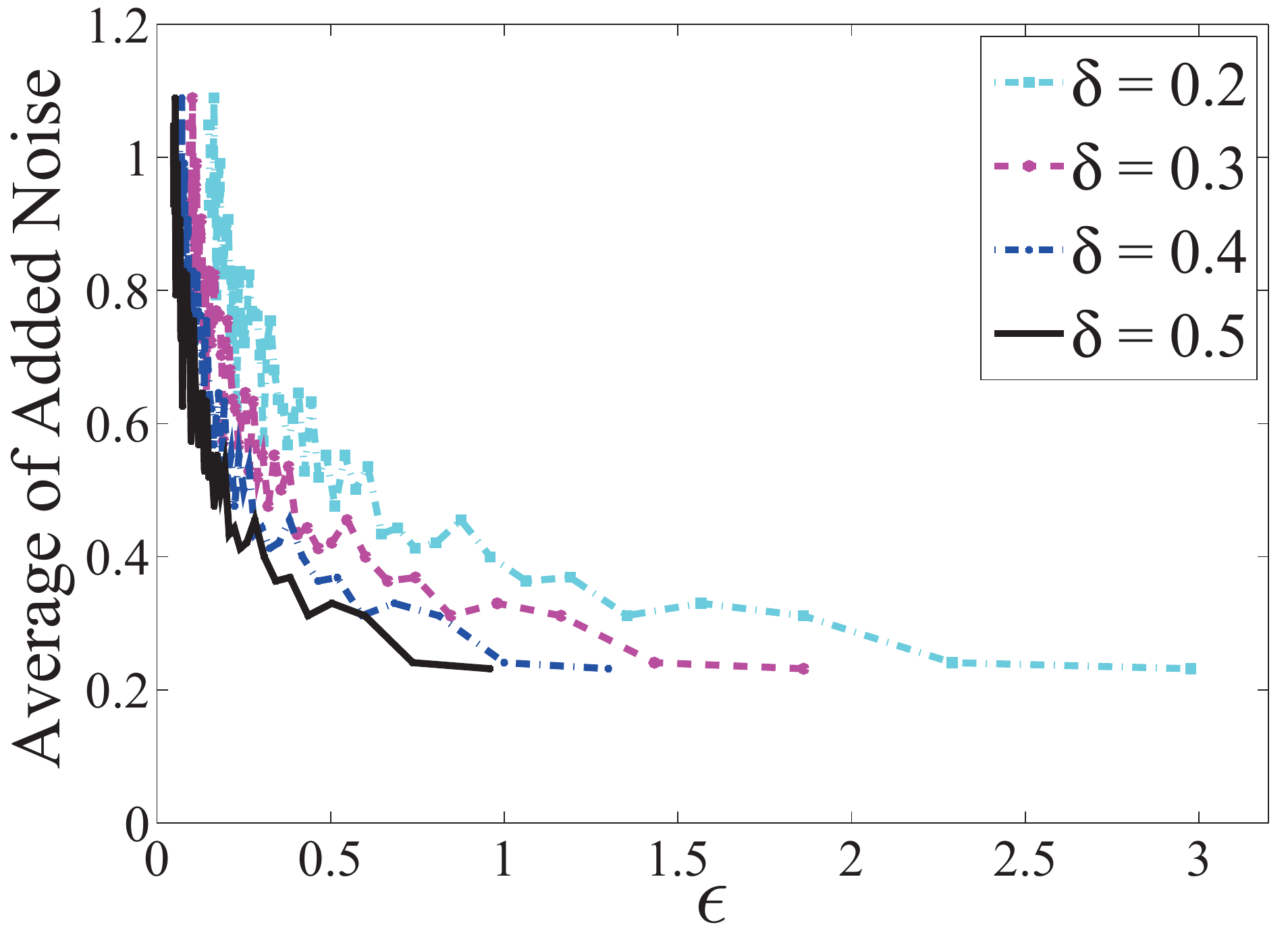}
}
\caption{Utility-Privacy Trade-off on Indoor Floorplan Dataset}
\label{trade_indoor}
\end{figure}

\smallskip \noindent \textsf{\textbf{Weight Comparison}}. As discussed, the advantage of the proposed privacy-preserving truth discovery mechanism in preserving good utility can be attributed to the weight estimation scheme. We illustrate this fact on the indoor floorplan dataset by comparing weights estimated from original and perturbed data. Figure \ref{weight_comparison} shows the estimated weights for 7 randomly selected users by the proposed method on original data and perturbed data using blue dotted lines. We obtain the groundtruth distance by measuring the hallway segments manually. This enables us to derive the true weight of each user for both cases, which are shown as black solid curves. By comparing true and estimated weights, we can observe the following phenomena: (1) The weights estimated by the proposed method are mostly consistent with the true weights, and thus weighted aggregation can outperform naive aggregation solutions such as mean or median in finding true information. (2) Compared with information quality on original data (Figure \ref{weight_clean}), we find that the $5$-th user adds large noise to protect his information, and thus on perturbed data, his weight is adjusted to a smaller value. This shows how the proposed mechanism can assign user weights based on user information quality, as explained in Section \ref{sec:method}. Correspondingly, the effect of added noise can be reduced during weighted aggregation, and thus aggregated result does not deviate much from the result before perturbation.
\begin{figure}[h]
\centering
\subfloat [Original Data]{
\includegraphics[width=0.21\textwidth]{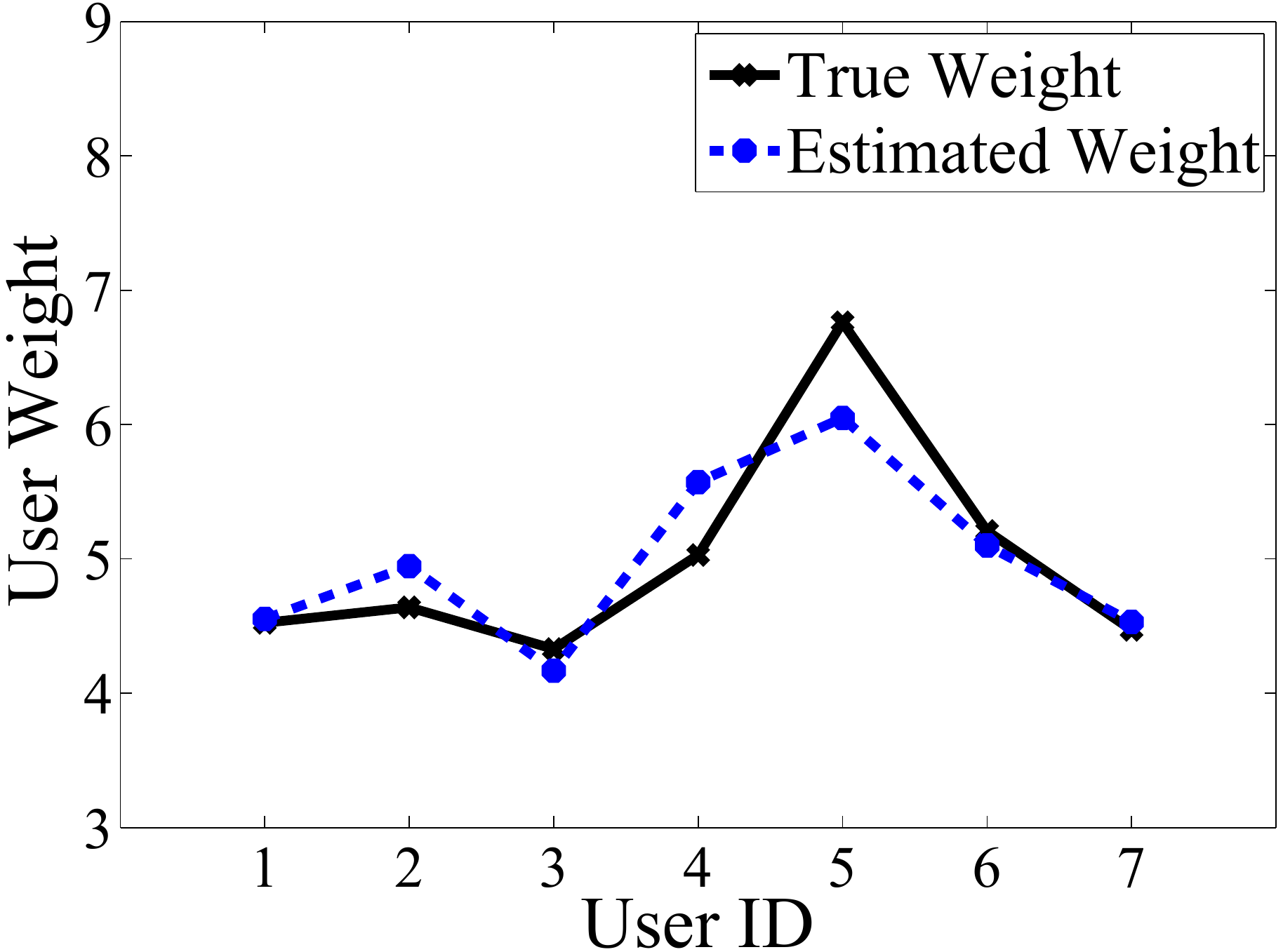}
\label{weight_clean}}
\hfill
\subfloat [Perturbed Data]{
\includegraphics[width=0.21\textwidth]{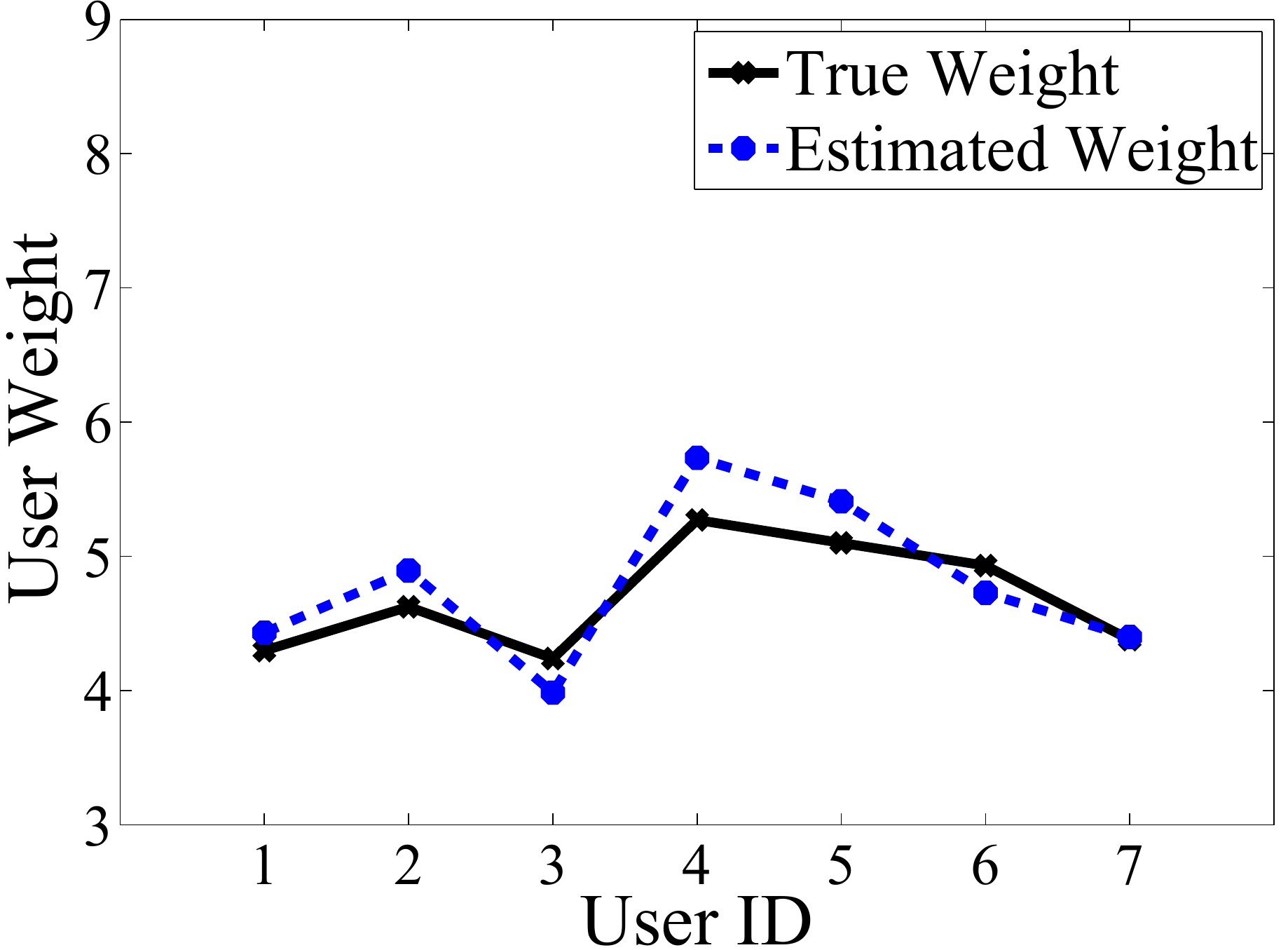}
\label{weight_noise}
}
\caption{Weight Comparison}
\label{weight_comparison}
\end{figure}

\subsection{Efficiency}
The last experiment shows the efficiency of the proposed mechanism. According to Algorithm \ref{alg:proposed}, the running time mainly comes from the execution of two parts, data perturbation and truth discovery procedure. Compared with time complexity of truth discovery, the time to add random noise is negligible, so we focus on analyzing the running time of truth discovery when different noise level is adopted.

Truth discovery is an iterative procedure whose running time is controlled by the number of iterations needed to achieve convergence. Existing literature has demonstrated that the running time of truth discovery increases linearly with respect to the number of objects \cite{crh_sigmod14} when the number of iterations is fixed, which is highly efficient. Therefore, in this experiment, we test the effect of noise level on running time, i.e., we check if the number of iterations is affected by noise level which leads to changes in running time. In practice, we set the convergence criterion for truth discovery in the following way: If the change in aggregated results is smaller than a threshold, the algorithm is terminated. We set the same threshold, vary the added noise, and record the running time of truth discovery on original and perturbed data. Figure \ref{running_time} reports the results in which the solid red line shows the running time of truth discovery on original data, and blue dots represent the running time on data with certain added noise. We can observe that running time after perturbation is slightly bigger than that on original data, but the running time does not change much when noise level varies. This shows that perturbation on user data does not change the running time of truth discovery approach, which guarantees practical deployment of the proposed mechanism on large-scale crowdsourcing applications.
\begin{figure}[tbh]
\centering
\includegraphics[width=0.28\textwidth]{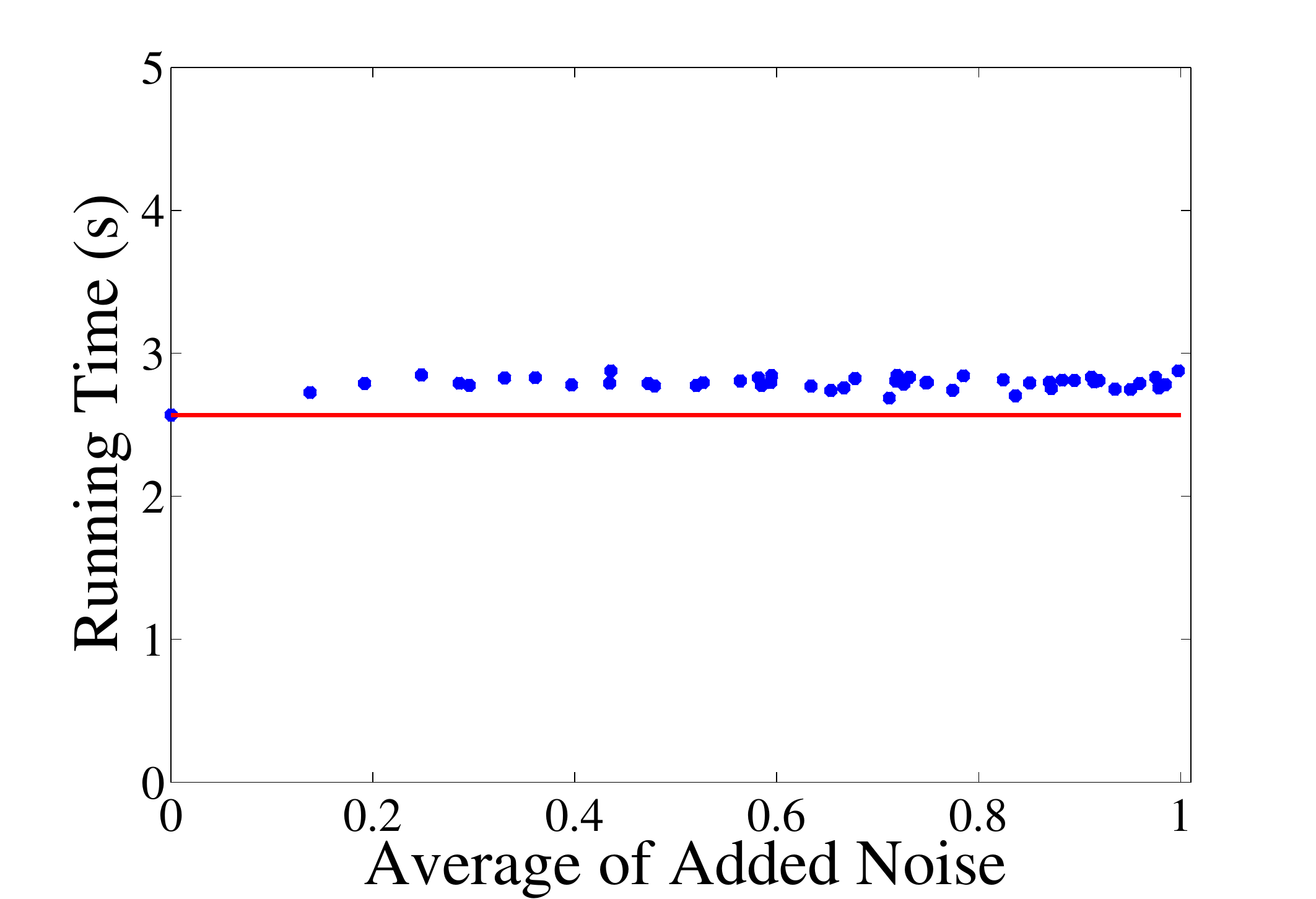}
\caption{Efficiency Study}
\label{running_time}
\end{figure}

\section{Related Work}
\label{sec:related}

Truth discovery has emerged as a hot topic for conflict resolution in data integration, and been applied in many other domains \cite{WKL+12,meng2015truth,liweighted}. By estimating source (user) quality from the data, the aggregated results are more reliable compared with naive solutions such as voting or averaging. Existing approaches include TruthFinder \cite{yin_kdd07}, AccuSim \cite{luna_survey_vldb12}, CRH \cite{crh_sigmod14}, etc. 
However, all these truth discovery approaches do not address the privacy concern in data collection. There are some recent work \cite{miao2015cloudpp,xu2016achieving,miao2017privacy,zheng2017privacy,zheng2018learning} which deal with this privacy concern based on encryption or secure multi-party computation techniques. Compared to them, the proposed mechanism in this paper provides a much more efficient perturbation based solution to privacy-preserving truth discovery.

Differential privacy \cite{dwork2006differential,kifer2012rigorous,mulle2015privacy,chen2011publishing,fan2014adaptive,alhadidi2012secure,machanavajjhala2016differential,goryczka2013secure,ying2013linear} is a quantified privacy definition for protecting sensitive information that needs to be released, and it balances the trade-off between privacy protection and utility loss. 
Among the related work of differential privacy, distributed differential privacy \cite{kasiviswanathan2011can,blum2005practical} shares some similarity with our work, and it enables individual information sources to add noise separately. However, in distributed differential privacy, the server is still assumed to be trusted and the protection is against information leakage to third parties via statistical queries. Hence their setting is different from the one in this paper. Another relevant topic is local differential privacy \cite{duchi2013local,dwork2013algorithmic,kairouz2014extremal} which deals with the scenario that users do not trust the server. In privacy analysis, we quantify the user privacy based on local differential privacy.

Among the related work on privacy-preserving data aggregation, some provide users with secure protocols that allow users to submit their sensitive information to a collector \cite{kajinopreserving,miao2015cloudpp,xu2016achieving,miao2017privacy,zheng2017privacy,zheng2018learning}. However, these methods are mainly based on encryption or secure multi-party computation, which requires expensive computation or communication. Therefore, none of them is an ideal solution to privacy-preserving truth discovery which usually involves a large number of users and thus requires efficient strategies.

On the other hand, some related work on privacy-preserving data aggregation are perturbation-based. These methods are designed for the computation of some statistics \cite{ganti2008poolview,zhang2012data}. They are not designed for truth discovery that automatically infers user weights from the data and conducts weighted aggregation. Thus these methods cannot be easily applied to privacy-preserving truth discovery.

Note that the aforementioned privacy-preserving data aggregation approaches deal with tasks that are different from truth discovery. Truth discovery automatically estimates user weights from the data and incorporates such weights in the truth computation. The iterative procedure of weight estimation and weighted aggregation steps in truth discovery make it quite different from other aggregation methods. Therefore, the proposed privacy-preserving truth discovery mechanism and analysis, which capture the unique characteristics of truth discovery task, differ from those in related work. The most relevant existing work is \cite{li2018efficient}, in which a privacy-preserving mechanism is proposed for truth discovery with categorical data, while in the paper, the proposed mechanism is for truth discovery with continuous data.

\section{Conclusions}
\label{sec:conclusion}
In order to extract reliable information from noisy crowdsourced sensory data, it is crucial to estimate the quality of individual users.
Truth discovery, in which the reliable information and user weights are inferred simultaneously, provides a nice way to mine such noisy sensory data. 
However, existing truth discovery methods fail to address the user privacy issue that arises in the data collection procedure.
In this paper, we propose a perturbation-based privacy-preserving truth discovery mechanism for crowd sensing systems. This mechanism is efficient and does not require any communication or coordination among mobile device users.
In this mechanism, each user samples a variance parameter ${\delta_s}^2$ from an exponential distribution and draws random noise from a Gaussian distribution with ${\delta_s}^2$ variance to perturb his data.
After collecting perturbed data, the server conducts weighted aggregation for final output.
As user weights can capture the information quality, the aggregated results on perturbed data do not differ much from the original aggregated values even when big noise is added.
We further analyze the performance of the proposed mechanism theoretically. We formally define $(\alpha,\beta)$-utility and $(\epsilon,\delta)$-privacy, and connect these concepts to the noise level $c$. The derived theorems show that larger noise leads to stronger privacy protection with less utility and vice versa.
We conduct experiments on not only synthetic datasets but also a crowdsourced indoor floorplan construction system. Results show that the proposed privacy-preserving truth discovery mechanism can tolerate big noise while the aggregation accuracy only drops slightly, which implies the guarantee of both good utility and strong privacy.
\appendix
\section{Special case}
For the special case where $c=1$, we have the following result in term of utility.
\begin{theorem}
	Let $c=1$, $\forall \alpha > \frac{15\sqrt{2\lambda_1}}{8}$,
	\begin{eqnarray}
	\lim_{S\to\infty} \Pr\{\frac{1}{N}\sum_{n=1}^{N}|x^*_n-\hat{x}^*_n| \geq \alpha\} = 0.
	\end{eqnarray}	\label{specialcase}
\end{theorem}
\begin{proof}
	When $c=1$, the distribution of noise variance is the same as the distribution of the error variance. Therefore, $Y^2_{s,s'} = \sigma^2_s+\sigma^2_{s'}+\delta^2_{s'}$ follows $\text{Gamma}(3,1/\lambda_1)$, with p.d.f. $h(y)=\frac{1}{2}\lambda_1^3y^2e^{-\lambda_1 y}$. It is easy to derive the p.d.f of $Y_{s,s'}$, which is $h'(y)=\lambda_1^3y^5e^{-\lambda_1 y^2}$. Moreover, $\mathbb{E}(Y) = \frac{15}{16}\sqrt{\lambda_1\pi}$ and $\mathbb{E}(Y^2) = \frac{3}{\lambda_1}$.
	
	Similar to the proof for Theorem 4.3, we have
	\begin{small}
	\begin{align}
	&\Pr\{\frac{1}{N}\sum_{n=1}^{N}|x^*_n-\hat{x}^*_n| \geq \alpha\}\leq4\sqrt{\frac{2}{\pi}}\frac{\frac{1}{S^2}\text{Var}(Y)}{(\alpha/2)^2}\notag\\
	&=4\sqrt{\frac{2}{\pi}}\frac{\frac{1}{S^2}(\mathbb{E}(Y^2) - \mathbb{E}^2(Y))}{(\alpha/2)^2}=\sqrt{\frac{2}{\pi}}\frac{48-12\lambda_1^2\pi}{S^2\alpha^2\lambda_1}.
	\end{align}
	\end{small}
	As $S$ goes to infinity, the right hand side tends to $0$. Therefore, $\forall \alpha > \frac{15\sqrt{2\lambda_1}}{8}$, we have $\lim_{S\to\infty} \Pr\{\frac{1}{N}\sum_{n=1}^{N}|x^*_n-\hat{x}^*_n| \geq \alpha\} = 0$. Thus Theorem \ref{specialcase} holds.
\end{proof}

\section{Proof of Lemma 4.4}
\begin{proof}
	To prove Eq.\ (\ref{lemma1eq}) is equivalent to prove the following inequality:
	\begin{small}
	\begin{eqnarray}
	S\sum_{s=1}^{S}w_st_s\leq\sum_{s=1}^{S}t_s\sum_{s'=1}^{S}w_{s'}.
	\end{eqnarray}
	\end{small}
	Moreover, we have:
	\begin{small}
	\begin{eqnarray}
	&&S\sum_{s=1}^{S}w_st_s - \sum_{s=1}^{S}t_s\sum_{s'=1}^{S}w_{s'}  \notag \\
    & = &S\sum_{s=1}^{S}w_st_s - \sum_{s=1}^{S}\sum_{s'=1}^{S}t_s w_{s'}\notag \\
	& = &(S - 1)\sum_{s=1}^{S}w_st_s - \sum_{s=1}^{S}\sum_{s'\neq s}^{S}t_{s} w_{s'} \notag \\
    & = &\sum_{s=1}^{\lceil \frac{S-1}{2} \rceil}\sum_{s'\leq s}(f(t_s) - f(t_{s'}))(t_s - t_{s'}).\notag
	\end{eqnarray}
    \end{small}
	According to the condition that $f$ is a monotonically decreasing function , we can obtain:
	\begin{displaymath}
	f(t_s) - f(t_{s'}) = \left\{\begin{array}{ll}
	\ \geq 0 &\quad \textrm{if $t_s - t_{s'} \leq 0$}\\
	\ \leq 0 &\quad \textrm{if $t_s - t_{s'} \geq 0$}
	\end{array} \right.
	\end{displaymath}
	It is obvious to see that for all $s$ and $s'$, if $s\neq s'$, the following inequality holds:
	\begin{eqnarray}
	(f(t_s) - f(t_{s'}))(t_s - t_{s'}) \leq 0.
	\end{eqnarray}
	Based on this observation, $\sum_{s=1}^{\lceil \frac{S-1}{2} \rceil}\sum_{s'\leq s}(f(t_s) - f(t_{s'}))(t_s - t_{s'})\leq 0$, which proves Eq.\ (\ref{lemma1eq}). Therefore, Lemma 4.4 holds.
\end{proof}

\end{document}